\documentclass[twocolumn]{arxiv_hb}

\usepackage{mathptmx}         
\usepackage{helvet}           
\usepackage{courier}          
\usepackage{type1cm}          
\usepackage{makeidx}          
\usepackage{graphicx}         
\usepackage{multicol}         
\usepackage[bottom]{footmisc} 
\usepackage{amssymb,amsmath}
\newtheorem{propo}{Proposition}
\DeclareMathAlphabet{\mathcal}{OMS}{cmsy}{m}{n} 

\usepackage{flushend}

\makeindex             

\begin{document}

\leadauthor{Behjat}

\title{{\Large Signal-Adapted Decomposition of Graph Signals}}

\shorttitle{Signal-Adapted Decomposition of Graph Signals}

\author[1,2,\Letter]{Harry H. Behjat}
\author[3]{Carl-Fredrik Westin}
\author[1,4]{Rik Ossenkoppele}
\author[5,6]{Dimitri Van~De~Ville}
\affil[1]{Department of Clinical Sciences Malm{\"o}, Lund University, Lund, Sweden}
\affil[2]{Department of Biomedical Engineering, Lund University, Lund, Sweden}
\affil[3]{Department of Radiology, Brigham and Women's Hospital, Harvard Medical School,Boston, Massachusetts, USA}
\affil[4]{VU University Medical Center, Amsterdam UMC, Amsterdam, Netherlands}
\affil[5]{Neuro-X Institute, Ecole Polytechnique Federale de Lausanne, Lausanne, Switzerland}
\affil[6]{Department of Radiology and Medical Informatics, University of Geneva, Geneva, Switzerland}
\affil[$\ast$]{\Letter~Corresponding author; e-mail: harry.behjat@gmail.com}
\date{}

\maketitle

\begin{abstract}\small Analysis of signals defined on complex topologies modeled by graphs is a topic of increasing interest. Signal decomposition plays a crucial role in the representation and processing of such information, in particular, to process graph signals based on notions of scale (e.g., coarse to fine). The graph spectrum is more irregular than for conventional domains; i.e., it is influenced by graph topology, and, therefore, assumptions about spectral representations of graph signals are not easy to make. Here, we propose a tight frame design that is adapted to the graph Laplacian spectral content of given classes of graph signals. The design is based on using the ensemble energy spectral density, a notion of spectral content of given signal sets that we determine either directly using the graph Fourier transform or indirectly through a polynomial-based approximation scheme. The approximation scheme has the benefit that (i) it does not require eigendecomposition of the Laplacian matrix making the method feasible for large graphs, and (ii) it leads to a smooth estimate of the spectral content. A prototype system of spectral kernels each capturing an equal amount of energy is initially defined and subsequently warped using the signal set's ensemble energy spectral density such that the resulting subbands each capture an equal amount of ensemble energy. This approach accounts at the same time for graph topology and signal features, and it provides a meaningful interpretation of subbands in terms of coarse-to-fine representations. We also show how more simplified designs of signal-adapted decomposition of graph signals can be adopted based on ensemble energy estimates. We show the application of proposed methods on the Minnesota road graph and three different designs of brain graphs derived from neuroimaging data.
\end{abstract}

\section{Introduction}
Many fields of science rely on network analysis to study complex systems. Networks are modeled mathematically as weighted graphs that have a set of nodes (vertices) with interactions between them represented by connections (links) and associated strengths. A rich repertoire of methods have been developed to pursue original queries and integrate the complexity of network structure into the analysis, subsequently providing new interpretations of datasets in divers scientific disciplines ranging from social sciences to physics and biology. One of the successes in network analysis is the ability to identify sets of nodes based on their connectivity. Traditional graph partitioning goes back to optimizing graph cuts~\citep{Fiedler1973}, while more recent community detection identifies sets of nodes that are more densely connected inside the set than outside~\citep{Girvan2002}. Community detection has been widely applied and many variants of the corresponding optimization criterion have been proposed~\citep{Newman2010}.

Another major trend is the emergence of methods for processing data that reside on top of networks~\citep{Ortega2018, Shuman2013, Sandryhaila2013, Stankovic2019}. Measurements on a network's nodes can be viewed as graph signals, enabling the extension of classical signal processing operations such as denoising, filtering, and transformation by leveraging the network's underlying connectivity. Numerous approaches have been proposed to generalize multiresolution transforms, filter bank designs, and dictionary constructions to the graph domain. These studies fall primarily into two categories: spatial (vertex) and spectral (frequency) designs. Schemes that fall within the former family include methods for designing wavelets for hierarchical trees \citep{Ram2, Gavish} and methods based on lifting schemes \citep{Jansen2009, Narang2009}. The latter family is based on spectral graph theory~\citep{Chung1997}, which is a powerful approach based on the eigendecomposition of matrices associated with graphs such as the adjacency matrix or graph Laplacian. Its strength lies in the global nature of eigenvectors, which capture key graph properties and enable solutions to convex relaxations of graph cut minimization~\citep{Luxburg2007}, or to define signal-processing operations by a graph equivalent of the Fourier transform \citep{Shuman2013}. In its application to graph signal processing, operations are performed in the spectral domain using graph spectral filters. One of the first proposals, the spectral graph wavelet transform (SGWT) frame \citep{Hammond2011}, is built using scaled cubic spline spectral kernels and a low-pass spectral polynomial kernel. Moreover, various constructions of systems of spectral graph kernels leading to tight frames were proposed in \citep{Leonardi2013, Shuman2020, Isufi2024}. Tight frames are particularly interesting because of their property of energy conservation between the original and transformed domain \citep{Benedetto}. Other approaches to spectral domain design include diffusion wavelets \citep{Coifman2006}, vertex-frequency frames \citep{Shuman2015acha, Stankovic2017, Ghandehari2021, Stankovic2020} and approaches to graph filter-bank design using bipartite graph decompositions \citep{Narang, Tanaka, Tay2023}, connected sub-graph decomposition \citep{Tremblay}, graph coloring \citep{ShumanPyramid} and Slepian functions that provide a trade-off between temporal and spectral energy concentration \citep{Vandeville2017}.

A key challenge of the graph spectrum is its dependence on the graph itself, making the spectral representation of a graph signal reliant on both the domain and the signal. However, most spectral designs define spectral windows independently of the graph and signal. The first major methodology that employed adaptation to the graph's spectral properties is the spectrum-adapted tight graph wavelet and vertex-frequency frames proposed in~\citep{Shuman2015ieee}. These kernels address the non-uniform distribution of Laplacian eigenvalues, ensuring each spectral kernel supports a similar number of eigenvalues. In \citep{Thanou2014,yankelevsky2016}, numerical dictionary learning approaches were proposed, where dictionaries are learned from training signals. These learned kernels are indirectly adapted to both the graph Laplacian spectrum and the training data by incorporating the graph structure into the learning process. An application-specific approach in \citep{Behjat2015} tailored the Meyer-like frame design \citep{Leonardi2013} to the spectral content of functional MRI signals, producing narrow-support kernels covering the lower end of the spectrum. Here we propose an approach for constructing tight graph frames that incorporate both the intrinsic graph topology, as in \citep{Shuman2015ieee}, and the characteristics of a given signal set. This is achieved by considering a graph-based energy spectral density notion that includes signal and topology properties and encodes the energy-wise significance of the graph eigenvalues. A system of spectral kernels tailored to the energy spectral density is  constructed by starting from the design of a prototype tight frame with uniform spectral coverage, followed by a warping step which incorporates the energy spectral density information to the prototype design, resulting in a tight frame with equi-energy subbands. We also present results of more simplified designs tight frames that are adapted to the signal energy content of signal sets at hand, without the explicit need to incorporate the warping step.      

\section{Preliminaries}
\label{sec:prelim}

\subsection{Graphs and Spectral Graph Theory}
A graph can be denoted as $G=(V,E)$ with $N_{g}$  vertices in set $V$, a set of edges as tuples ($i,j$) in $E$ where $i,j \in V$. The size of the graph is the number of vertices. In this work, we only consider undirected graphs without self-loops. Algebraically, $G$ can be represented with the node-to-node adjacency matrix $A$, with elements $a_{i,j}$ denoting the weight of the edge $(i,j)$ if $(i,j) \in  E$; $a_{i,j} = 0$ if $(i,j) \notin  E$. The degree matrix $D$ of $G$ is diagonal with elements $ d_{i,i}=\sum_{j} a_{i,j}$. The Laplacian matrices of $G$ in combinatorial form $L$ and normalized form $\mathcal{L}$ are defined as 
\begin{align}
L & = D-A,  \\ 
\mathcal{L} & =  D^{-1/2} L D^{-1/2},
\end{align} 
respectively. Both $L$ and $\mathcal{L}$ are symmetric and positive semi-definite, and thus, their diagonalizations lead to a set of $N_{g}$ real, non-negative eigenvalues that define the graph Laplacian spectrum
\begin{equation}
\label{eq:graphSpectrum}
\Lambda(G) = \{ 0 = \lambda_{1} \le \lambda_{2} \cdots \le \lambda_{N_{g}} = \lambda_{\textrm{max}} \}.
\end{equation}
The corresponding set of eigenvectors $\{\vec{\chi}_{l}\}_{l=1}^{N_{g}}$ forms a complete set of orthonormal vectors that span the graph spectral domain \citep{Chung1997}. When necessary, we use the notations $\Lambda_{L}(G)$ and $\Lambda_{\mathcal{L}}(G)$ to distinguish between the two definitions of the graph Laplacian. As the eigenvalues may be repetitive, for each $\lambda_{l}$, we denote its algebraic multiplicity by $m_{\lambda_{l}}$ and the index of its first occurrence by $i_{\lambda_{l}}$. That is, if $\lambda_{l}$ is singular, i.e. $m_{\lambda_{l}} =1$, then $i_{\lambda_{l}} = l$, and if $\lambda_{l}$ is repetitive, then $i_{\lambda_{l}}  \le l$. The multiplicity of eigenvalues equal to zero reflects the number of connected components in the graph. In this paper, only connected graphs are considered, and thus, $m_{\lambda_{1}} = 1$.  

\subsection{Graph Signals: Vertex versus Spectral Representations} 
\label{sec:graphFourierTrans}
Let $\ell_{2}(G)$ denote the Hilbert space of all square-summable real-valued vectors $\vec{f}\in\mathbb{R}^{N_{g}}$, with the inner product defined as 
\begin{equation}
\langle \vec{f}_{1}, \vec{f}_{2} \rangle  = \sum_{n=1}^{N_g} f_{1}[n] f_{2}[n], \quad \forall \vec{f}_{1}, \vec{f}_{2} \in \ell_{2}(G) 
\end{equation}
and the norm as  
\begin{equation}
{|| \vec{f} ||}_{2}^{2} =  \langle \vec{f}, \vec{f} \rangle = \sum_{n=1}^{N_g} |f[n]|^{2}, \quad \forall \vec{f} \in \ell_{2}(G).
\end{equation}
\\
A real signal defined on the vertices of a graph, $\vec{f}: V \rightarrow \mathbb{R}$, can be seen as a vector in $\ell_{2}(G)$, where the $n$-th element represents the value of the signal on the $n$-th vertex. 
 
For any $\vec{f} \in \ell_{2}(G)$, its spectral representation $\hat{\vec{f}} \in \ell_{2}(G)$, known as the graph Fourier transform of $\vec{f}$, can be used to express $\vec{f}$ in terms of the graph Laplacian eigenvectors
\begin{equation}
\label{eq:inverseFourier}
f[n] =  \sum_{l=1}^{N_{g}} \underbrace{\langle \vec{f}, \vec\chi_{l} \rangle}_{=\hat{f}[l]} \chi_{l}[n].
\end{equation}

With this definition of the Fourier transform, it can be shown that the Parseval relation holds \citep{Shuman2015acha}  
\begin{equation}
\label{eq:parseval}
\forall \vec{f}_{1}, \vec{f}_{2} \in \ell_{2}(G), \quad \langle \vec{f}_{1}, \vec{f}_{2} \rangle = \langle \hat{\vec{f}}_{1}, \hat{\vec{f}}_{2} \rangle.
\end{equation}

\subsection{Filtering of Graph Signals}
In the graph setting, the generalized convolution product is defined as 
\begin{equation}
\label{eq:generalizedConvolution}
( \vec{f}_{1} \ast \vec{f}_{2})[n] = \sum_{l=1}^{N_g} \hat{f}_{1}[l] \hat{f}_{2}[l] \chi_{l}[n], \quad \forall \vec{f}_{1}, \vec{f}_{2} \in \ell_{2}(G)). 
\end{equation}

In analogy with conventional signal processing, filtering of graph signals can be viewed as an operation in the spectral domain. For a given graph signal $\vec{f} \in \ell_{2}(G)$ and graph filter $\vec{g} \in \ell_{2}(G)$, defined through its Fourier transform $\hat{\vec{g}}$, the filtered signal, denoted by $(F_{\vec{g}} \vec{f})$, can be obtained as   
\begin{align}
\label{eq:filteringVertexDomain}
(F_{\vec{g}} \vec{f})[n] & = (\vec{g} \ast \vec{f})[n]  
\\ 
& 
\label{eq:filteringVertexDomain2}
\stackrel{(\ref{eq:generalizedConvolution})}{=} \sum_{l=1}^{N_g} \hat{g}[l] \hat{f}[l] \chi_{l}[n].
\end{align}
For the graph filter $\vec{g}$, the filter response of an impulse at vertex $m$ 
\begin{equation}
\label{eq:impulse}
\vec{f} = \vec{\delta}_{m} \leftrightarrow \hat{\delta}_{m}[l] = \langle \vec{\delta}_{m}, \vec{\chi}_{l} \rangle = \chi_{l}[m], 
\end{equation}
can be obtained as
\begin{equation}
\label{eq:impulseFiltering}
(F_{\vec{g}} \vec{\delta}_{m})[n] 
= 
\sum_{l=1}^{N_g} \hat{g}[l] \chi_{l}[m] \chi_{l}[n].
\end{equation}
The impulse response of a graph filter is, in general, shift-variant; i.e, the impulse response at one vertex is not simply a shifted version of the impulse response at any other node. This is due to the absence of a well-defined shift operator in the graph setting as that defined in the Euclidean setting. Therefore, a graph filter is conventionally defined by its spectral kernel $\hat{\vec{g}}$ rather than by its impulse response.

Although the graph spectrum is discrete, to design spectral kernels, it is often more elegant to define an underlying smooth continuous kernel. Let $L_{2}(G)$ denote the Hilbert space of all square-integrable spectral functions $K(\lambda): [0,\lambda_{\textrm{max}}] \rightarrow  \mathbb{R}^{+}$, with the inner product defined as 
\begin{equation}
\label{eq:L2innerproduct}
\left  \langle K_{1},K_{2}  \right  \rangle_{\textrm{L}_{2}} = \int_{-\infty}^{+\infty} K_{1}(\lambda) K_{2}(\lambda) d\lambda, \quad \forall K_{1},K_{2} \in \textrm{L}_{2}(G),
\end{equation}
and the $\textrm{L}_{2}$-norm defined as
\begin{equation}
\| K\|_{\textrm{L}_{2}}^{2} =  \left  \langle K,K  \right  \rangle_{\textrm{L}_{2}}, \quad \forall K \in \textrm{L}_{2}(G).
\end{equation}
A discrete version of $K(\lambda) \in L_{2}(G)$ can then be determined as 
\begin{equation}
\label{eq:kernelSamplig}
k[l] = K(\lambda_{l}), \quad l=1,\ldots, N_{g}.
\end{equation} 
Note that although $\vec{k}$ is defined in the spectral domain, it is not linked to any explicit vertex representation, and thus, the Fourier symbol $\:\:\: \hat{} \:\:\:$ is not used for their denotation. This notation convention will be used throughout the paper.

\subsection{Dictionary of Graph Atoms}
For a given spectral kernel $\vec{k}$ associated with $K(\lambda)$, the vertex-domain impulse responses are obtained as
\begin{equation}
\label{eq:atoms}
{\vec\psi}_{K,m} = (F_{\vec{k}} \vec{\delta}_{m}) 
\leftrightarrow \hat{{\vec\psi}}_{K,m}[l] = k[l] \chi_{l}[m].
\end{equation}
The collection of impulse responses $\{ {\vec\psi}_{K,m}\}_{m=1}^{N_g}$ are considered as graph \textit{atoms} associated with  spectral kernel ${K}(\lambda)$. 
Given a set of $J$ spectral kernels $\{\vec{k}_{j} \in \ell_{2}(G)\}_{j=1}^{J}$, a dictionary of graph atoms $D_{G}$ with $J N_{g}$ elements can be obtained 
\begin{equation}
\label{eq:dictionary}
D_{G} = \Big \{ \{{\vec\psi}_{K_{j},m}\}_{j=1}^{J}\Big \}_{m=1}^{N_{g}}.
\end{equation}
The atoms of $D_{G}$ form a frame in $\ell_{2}(G)$ if there exist bounds $B_2 \ge B_1 > 0$ such that \citep{Benedetto} 
\begin{equation}
\label{eq:frame}
\forall  \vec{f} \in \ell_{2}(G), \quad B_1 ||\vec{f}||_{2}^{2} \le \sum_{j,m}  | \langle \vec{f},{\vec\psi}_{K_{j},m} \rangle |^{2} \le B_2 ||\vec{f}||_{2}^{2},
\end{equation}
where the frame bounds are given by 
\begin{align}
\label{eq:G1}
B_1  = \min_{\lambda \in [0,\lambda_{\textrm{max}}]} G(\lambda), \quad B_2  = \max_{\lambda \in [0,\lambda_{\textrm{max}}]} G(\lambda),
\end{align}
and $G(\lambda)\in L_{2}(G)$ is defined as  
\begin{equation}
\label{eq:G2}
G(\lambda) = \sum_{j=1}^{J} | K_{j}(\lambda)|^{2}.
\end{equation}
In particular, $D_{G}$ forms a tight frame if 
\begin{equation}
\label{eq:tightFrame}
\forall \lambda \in [0,\lambda_{\textrm{max}}], \quad G(\lambda) = C,
\end{equation}
and a Parseval frame if $C=1$. 

\subsection{Decomposition of Graph Signals}
\subsubsection*{Direct Decomposition}
To decompose a graph signal $\vec{f}$ onto a set of the atoms in $D_{G}$, the coefficients can be obtained as 
\begin{align}
\label{eq:coeffs1}
	c_{K_{j},m} & = \langle \vec{f},{\vec\psi}_{K_{j},m} \rangle 
	\\
	\label{eq:coeffs2}
	 &\stackrel{(\ref{eq:parseval})}{=} \sum_{l=1}^{N_g} \hat{\psi}_{K_{j},m}[l] \hat{f}[l], 
\\
\label{eq:coeffs3}
	 & \stackrel{(\ref{eq:atoms})}{=} \sum_{l=1}^{N_g} {k}_{j}[l] \hat{f}[l] \chi_{l}[m].
\end{align}
Relation (\ref{eq:coeffs3}) shows that the direct decomposition requires a full eigendecomposition of the $L$ since it requires i) the Laplacian eigenvectors $\{\vec\chi_{l}\}_{l=1}^{N_{g}}$ and ii) the graph Fourier transform of the signal $\hat{\vec{f}}$. 

If $D_{G}$ forms a Parseval frame, the coeficents can be used to recover the original signal as
\begin{align}
f[n] & = \sum_{j} \sum_{m} c_{K_{j},m} {\vec\psi}_{K_{j},m}
\nonumber
\\
&  
= 
\sum_{j} \sum_{m} \sum_{l} {k}_{j}[l] \hat{f}[l] \chi_{l}[m] \sum_{l^{'}} k_{j}[l^{'}] \vec\chi_{l^{'}}[m] \chi_{l^{'}}[n]
\nonumber
\\
& = \sum_{l} \sum_{l^{'}} \sum_{j} {k}_{j}[l] {k}_{j}[l^{'}] \hat{f}[l] \chi_{l^{'}}[n]  \underbrace{\sum_{m} \chi_{l}[m] \chi_{l^{'}}[m]}_{\delta_{l-l^{'}}}
\nonumber
\\
& = \sum_{l} \underbrace{\sum_{j} {k}_{j}^{2}[l]}_{= 1} \hat{f}[l] \chi_{l}[n].
\end{align}

\subsubsection*{Decomposition Through Polynomial Approximation}
\label{sec:decomposition_approx}
The decomposition of $\vec{f}$ on $D_{G}$ leads to a coefficient vector associated to each $\vec{k}_{j}$ given as
\begin{align}
\label{eq:coeffVector1}
\vec{c}_{K_{j}} & \: \:\: = [c_{K_{j},1}, c_{K_{j},2}, \ldots, c_{K_{j},N_{g}}]^{T} 
\\
\label{eq:coeffVector2}
 & \stackrel{(\ref{eq:coeffs3})}{=} \sum_{l=1}^{N_g} {k}_{j}[l] \hat{f}[l] \vec\chi_{l},  
\end{align}
that can be interpreted as filtered versions of $\vec{f}$ with different spectral kernels $\{ \vec{k}_{j}\}_{j=1}^{J}$. Due to the redundancy of such a transform, it is beneficial to implement the transform using a fast algorithm, rather than using the explicit computation of the coefficients through (\ref{eq:coeffs3}). Moreover, for large graphs, it can be cumbersome to compute the full eigendecomposition of $L$, and in extensively large graphs---unless for graphs of special form e.g. those that have structures similar to graphons~\citep{Ghandehari2022}---this can in fact be impossible. One solution to overcome this computational burden is to use a polynomial approximation scheme.   

One such algorithm is the truncated Chebyshev polynomial approximation method \citep{Hammond2011}, which is based on considering the expansion of the continuous spectral window functions $\{K_{j}(\lambda)\}_{j=1}^{J}$ with the Chebyshev polynomials $C_{p}(x) = \cos(p \arccos (x))$ as  
\begin{equation}
\label{eq:chebyExpansion}
	K_{j}(\lambda)= \frac{1}{2}{d}_{K_{j},0} + \sum_{p=1}^{\infty} {d}_{K_{j},p} \:\bar{C}_{p} \left (\lambda \right ), 
\end{equation}
where $\bar{C}_{p}(x) = C_{p}(\frac{x - b} {b})$, $b= \lambda_{\textrm{max}}/2$ and ${d}_{K_{j},p}$ denote the Chebyshev coefficients obtained as
\begin{equation}
\label{eq:coeffsApprox}
	d_{K_{j},p} = \frac{2}{\pi} \int_{0}^{\pi} \cos(p \theta) K_{j}(b(\cos (\theta) +1)) d\theta. 
\end{equation}
By truncating (\ref{eq:chebyExpansion}) to $M$ terms, $K_{j}(\lambda)$ can be approximated as an $M$-th order polynomial $P_{j}(\lambda) \in L_{2}(G)$. Consequently, $\vec{c}_{K_{j}}$ can be approximated as 
\begin{align}
\label{eq:coeffsApprox}
	{\vec{c}}_{K_j} 
&\stackrel{(\ref{eq:coeffVector2})}{=} \sum_{l=1}^{N_g} \underbrace{k_{j}[l]}_{K_{j}(\lambda_{l})} \hat{f}[l] \vec\chi_{l} 	
\\ &  \approx \sum_{l=1}^{N_g} P_{j}(\lambda_{l}) \hat{f}[l] \vec\chi_{l} 	
\\ \label{eq:polyLaplacian} &  = P_{j}(L) \sum_{l=1}^{N_g} \hat{f}[l] \vec\chi_{l} 	
\\ & \stackrel{(\ref{eq:inverseFourier})}{=} P_{j}(L)   \vec{f}
\end{align}
where in (\ref{eq:polyLaplacian}) we exploit the property $L \chi_{l} = \lambda_{l} \chi_{l} \Rightarrow P_{j}(L) \chi_{l} = P_{j}(\lambda_{l}) \chi_{l} $.  

\section{Spectral Warping}
Given a spectral kernel $K(\lambda) \in L_{2}(G)$ and a function $T_{F}(\lambda) \in L_{2}(G): [0,\lambda_{\textrm{max}}] \rightarrow [0,\lambda_{\textrm{max}}]$, $K(\lambda)$ can be warped using $T_{F}(\lambda)$ as $K'(\lambda) = K(T_{F}(\lambda))$. The resulting kernel $K'(\lambda)$ remains a kernel $\in L_{2}(G)$ and also satisfies $\| K(T_{F}(\lambda))\|_{\textrm{L}_{2}} = \| K\|_{\textrm{L}_{2}}$.

More generally, given a system of spectral kernels $\{K_j(\lambda)\}_{j=1}^{J}$ that form a tight frame, spectral warping using $T_{F}(\lambda)$ results in a warped system of kernels which retain the tight frame property, i.e. (\ref{eq:G1})-(\ref{eq:tightFrame}) for the same frame bounds as in the original system of kernels, since:  
\begin{align}
\sum_{j=1}^{J} \vert {K'}_{j}(\lambda)\vert ^{2} 
& = \sum_{j=1}^{J} \vert K_{j}(\underbrace{T_{F}(\lambda)}_{:=\lambda^{'}})\vert ^{2}, \quad \forall \lambda \in  [0,\lambda_{\text{max}}] 
\nonumber
\\
& \: \: = \sum_{j=1}^{J} \vert K_{j}(\lambda^{'})\vert ^{2}, \quad \forall \lambda^{'} \in  [0,\lambda_{\text{max}}].
\nonumber
\end{align}

\section{Prototype Systems of Spectral Kernels}
There are a number of prototype systems of spectral kernels that can be leveraged in different stages of performing signal-adapted decomposition of graph signals. In particular, there are two stages that systems of spectral kernels are used: (i) to estimate the energy spectral density of a graph signal; (ii) to decompose a graph signal. For (i), generally a system with a large number of kernels is desired to obtain a suitable and sufficiently detailed estimate of the distribution of signal energy across the graph spectrum. For (ii), generally a small number of kernels is desired to characterise a the signal at multiple interpretable spectral subbands. In the following, the construction detail and properties of two different categories of kernels are given, the first of which i used for (i) and the second one is used for (ii).

\subsection{B-spline based System of Spectral Kernels} 
\label{sec:bsplineBasedSystem}
The central B-spline of degree $n$, denoted $\beta^{(n)}(x)$, is a compactly-supported function in the interval $[-\Delta^{(n)},\Delta^{(n)}]$, i.e., 
$\beta^{(n)}(x)=0$ for all $|x| \ge \:  \Delta^{(n)}$ where $\Delta^{(n)}=(n+1)/2$, and is obtained through the $(n+1)$-fold convolution as
\begin{equation}
\beta^{(n)}(x) = \underbrace{\beta^{(0)}(x) \ast\beta^{(0)}(x) \ast \cdots \beta^{(0)}(x)}_{(n+1) \textrm{times}},
\end{equation}
where  
\begin{equation}
\beta^{(0)}(x) = 
\begin{cases}
1, \quad -\frac{1}{2}<x<\frac{1}{2}
\\
\frac{1}{2}, \quad |x| = \frac{1}{2}
\\
0, \quad \textrm{otherwise}.
\end{cases}
\end{equation}

\begin{propo} 
\label{propo:BsplineFrame}
(B-spline based Parseval Frame on Graphs) 
For a given graph $G$ and B-spline generating function $\beta^{(n)}(x)$, $n\ge 2$, a set of B-spline based spectral kernels $\{B_{j}(\lambda) \in L_{2}(G)\}_{j=1}^{J}$ can be defined as
\begin{equation}
\label{eq:B}
B_{j}(\lambda) =
\begin{cases}
\widetilde{B}_{j}(\lambda) + \sum_{i=-\Delta}^{0}\widetilde{B}_{i}(\lambda), j = 1 
\\
\widetilde{B}_{j}(\lambda), j = 2,\ldots,J-1 
\\
\widetilde{B}_{j}(\lambda) + \sum_{i=J+1}^{J+\Delta+1}\widetilde{B}_{i}(\lambda), j = J 
\end{cases}
\end{equation}
where $\Delta=\lfloor n/2 \rfloor -1$ and $\widetilde{B}_{\cdot}(\lambda) \in L_{2}(G)$ is defined as
\begin{equation}
\label{eq:B_l}
\widetilde{B}_{l}(\lambda) = \sqrt{\beta^{(n)}\left (\frac{\lambda_{\textrm{max}}}{J-1} (\lambda - l+1) \right )}, \quad l=-\Delta, \ldots, J+\Delta+1.
\end{equation}
 
The system of kernels $\{B_{j}(\lambda)\}_{j=1}^{J}$ satisfy 
\begin{equation}
\sum_{j=1}^{J} |B_{j}(\lambda)|^{2} = 1, \quad \forall \lambda \in [0,\lambda_{\textrm{max}}],
\label{eq:pouBsplines}
\end{equation}
and, thus, their associated dictionary of atoms forms a Parseval frame.
\end{propo}
\begin{proof}
See Appendix~I.
\end{proof}
Fig.~\ref{fig:ust} shows two realizations of B-spline based systems of spectral kernels. The systems entail wide, overlapping passband kernels. Moreover, the kernels are smooth, which enables their seamless approximation as low order polynomials.
  
\begin{figure}[]
\centering
\includegraphics[width=.48\textwidth]{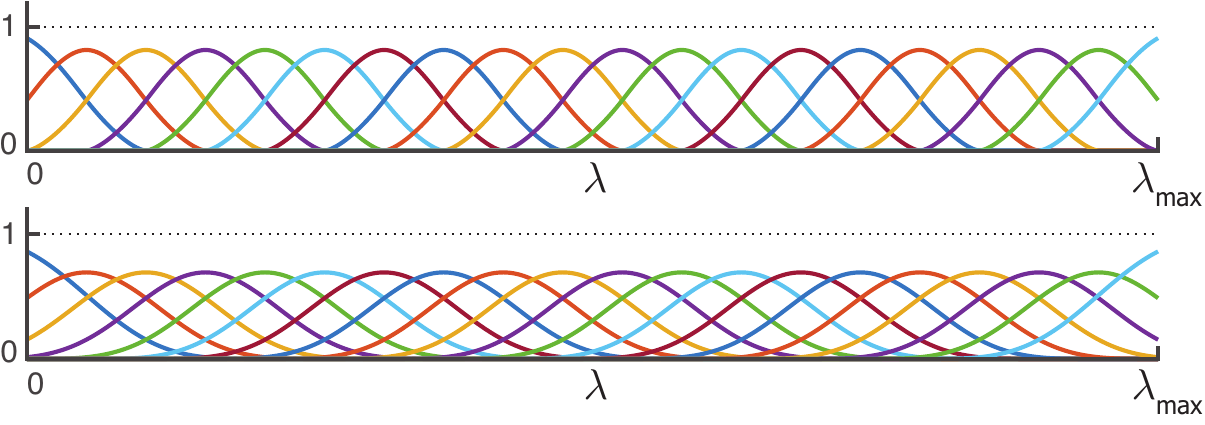}
\caption[]{B-spline based system of spectral kernels with $20$ spectral bands constructed based on B-splines of order 3 (top) and 7 (bottom). The dotted lines correspond to $G(\lambda)$ in (\ref{eq:G2}), i.e., property~(\ref{eq:pouBsplines}).}
\label{fig:ust}
\end{figure}

The system of kernels defined in Proposition~\ref{propo:BsplineFrame} are uniformly spread across the spectrum. Such a system is suitable to estimate the overall of distribution of energy of graph signals across the spectrum. However, given that real-world graph signals are generally lowpass, it can be efficient to leverage a larger number of kernels at the lower-end of the spectrum to enable a more fine-scale estimation of the low-frequency energy content. To this aim, the uniform system of kernels given in Proposition~\ref{propo:BsplineFrame} can be tailored to entail a multiresolution characteristic, wherein the lower-end kernels become more narrow-band passband up to a desired pivot point---$\lambda^{(\text{piv})} \in (0, \lambda_{\textrm{max}}), ~ \lambda^{(\text{piv})} \ll \lambda_{\textrm{max}}$---in the spectrum after which the kernels become more wide-band passband kernels. This can be done by defining a piecewise linear warping function as 
\begin{align}
P(\lambda) & = 
\begin{cases} 
m_1 \lambda, &~~~~~~~ 0 \leq \lambda \leq \lambda^{(\text{piv})}, \\ 
m_2(\lambda - \lambda^{(\text{piv})}) + y^{(\text{piv})}, & \lambda^{(\text{piv})} < \lambda \leq \lambda_{\text{max}},
\end{cases}
\end{align}
where $m_1 = y^{(\text{piv})}/\lambda^{(\text{piv})}$ and $m_2 = (1 - y^{(\text{piv})})/(\lambda_{\text{max}}-\lambda^{(\text{piv})})$ are line slopes, $y^{(\text{piv})} = N_{\textrm{lower}} / N_{\textrm{total}}$, and $N_{\textrm{lower}}$ and $N_{\textrm{total}}$ are the number of kernels that are to cover the lower-end of the spectrum and the entire spectrum, respectively. Smoothing $P(\lambda)$ eliminates the sharp transition between its piecewise linear components, yielding a well-behaved warping function. This ensures a smooth system of warped kernels, enabling practical polynomial approximations for signal decomposition, cf.~\ref{sec:decomposition_approx}. Fig.~\ref{fig:sosks57} shows a multiresolution system of B-spline based kernels obtained by warping the uniformly distributed system of B-spline based kernels in Proposition~\ref{propo:BsplineFrame}, with $N_{\textrm{total}} = 57$, using a smoothed version of $P(\lambda)$. 

\begin{figure*}[]
\centering
\includegraphics[width=.65\textwidth]{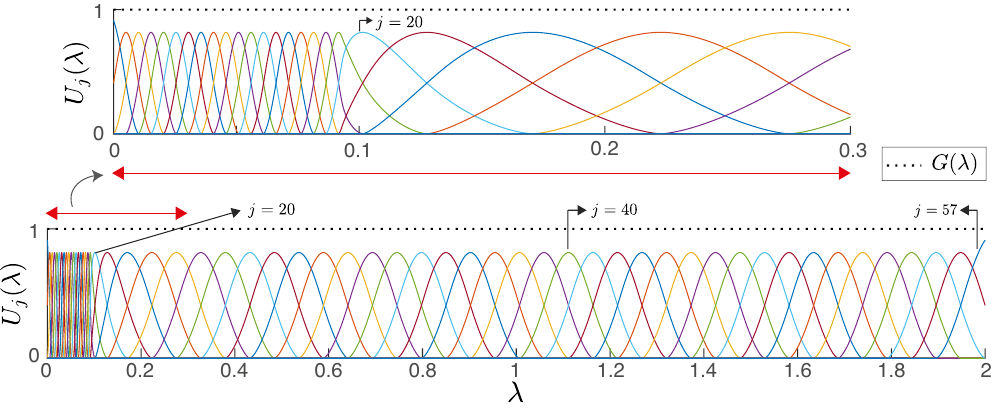}
\caption[]{SOSKS57. The dotted lines correspond to $G(\lambda)$ in (\ref{eq:G2}).} 
\label{fig:sosks57}
\end{figure*}

\subsection{Meyer-like Uniform System of Spectral Kernels}
\label{sec:uniformMeyer}
\begin{propo} 
\label{propo:umt} (uniform Meyer-type (UMT) system of spectral kernels) Using the auxiliary function of the Meyer wavelet, given by \citep{Meyer}  
\begin{equation}
\label{eq:nu}
\nu(x) = x^4(35-84x+70x^2-20x^3),
\end{equation} 
a set of $J \ge 2$ spectral kernels defined as 
\begin{subequations}
\label{eq:uniformFrame}
\begin{align}
U_{1}(\lambda) & = 
\begin{cases}
\label{eq:kernelScale1}
1 &   \quad \quad \quad \:  \forall \lambda  \in [0,a]  \\
\cos(\frac{\pi}{2} \nu (\frac{1}{\gamma-1}(\frac{\lambda}{a}-1))) & \quad \quad \quad  \: \forall \lambda \in ]a,\gamma a]  \\
0 &   \quad \quad \quad \: \emph{elsewhere}
\end{cases}
\\ U_{j}(\lambda) & = 
\begin{cases}
\label{eq:kernelScalej}
\sin(\frac{\pi}{2} \nu (\frac{1}{\gamma-1}(\frac{\lambda - (j-2)\Delta }{a}-1)))  &  \forall \lambda \in ]\lambda_{\textrm{I}},\lambda_{\textrm{II}}]    \\
\cos(\frac{\pi}{2} \nu (\frac{1}{\gamma-1}(\frac{\lambda - (j-1) \Delta}{a}-1))) &  \forall \lambda \in ]\lambda_{\textrm{II}},\lambda_{\textrm{II}} + \Delta]  \\
0 & \emph{elsewhere}
\end{cases}
\\ U_{J}(\lambda) &= 
\begin{cases}
\label{eq:kernelScaleJ}
\sin(\frac{\pi}{2} \nu (\frac{1}{\gamma-1}(\frac{\lambda - (J-2) \Delta }{a}-1)))  &  \forall \lambda \in ]\lambda_{\textrm{I}},\lambda_{\textrm{II}}]    \\
1 &  \forall \lambda \in ]\lambda_{\textrm{II}},\lambda_{\textrm{II}} + a]  \\
0 & \emph{elsewhere}
\end{cases}
\end{align}
\end{subequations}
can be constructed, where 
\begin{subequations}
\begin{align}
\label{eq:Delta}
\Delta & = \gamma a - a, \\ 
\label{eq:Lambda}
\lambda_{\textrm{I}}  & = \: \:a + (j-2) \Delta, \\
\lambda_{\textrm{II}}  & = \gamma a+ (j-2) \Delta, \\
\label{eq:a}
a &  = \frac{\lambda_{\textrm{max}}}{J\gamma -J -\gamma +3}. 
\end{align}
\end{subequations}
Fig.~\ref{fig:uniformFrameConstruction} illustrates the notations used. By setting $\gamma = 2.73$, the set of kernels defined in (\ref{eq:uniformFrame}) satisfies the uniformity constraint given in (\ref{eq:uniformityConstraintContinious}). The atoms of a dictionary constructed using this set of spectral kernels form a Parseval frame on $\ell_{2}(G)$.
\end{propo}

\begin{proof}
See Appendix~II.
\end{proof}
    
\begin{figure}[]
\centering
\includegraphics[width=.48\textwidth]{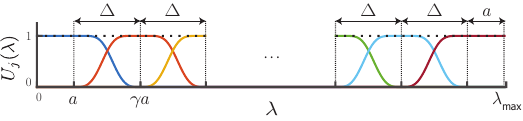} 
\caption[]{Construction of UMT system of spectral kernels.}
\label{fig:uniformFrameConstruction}
\end{figure}

\begin{figure}[]
\centering
\includegraphics[width=.48\textwidth]{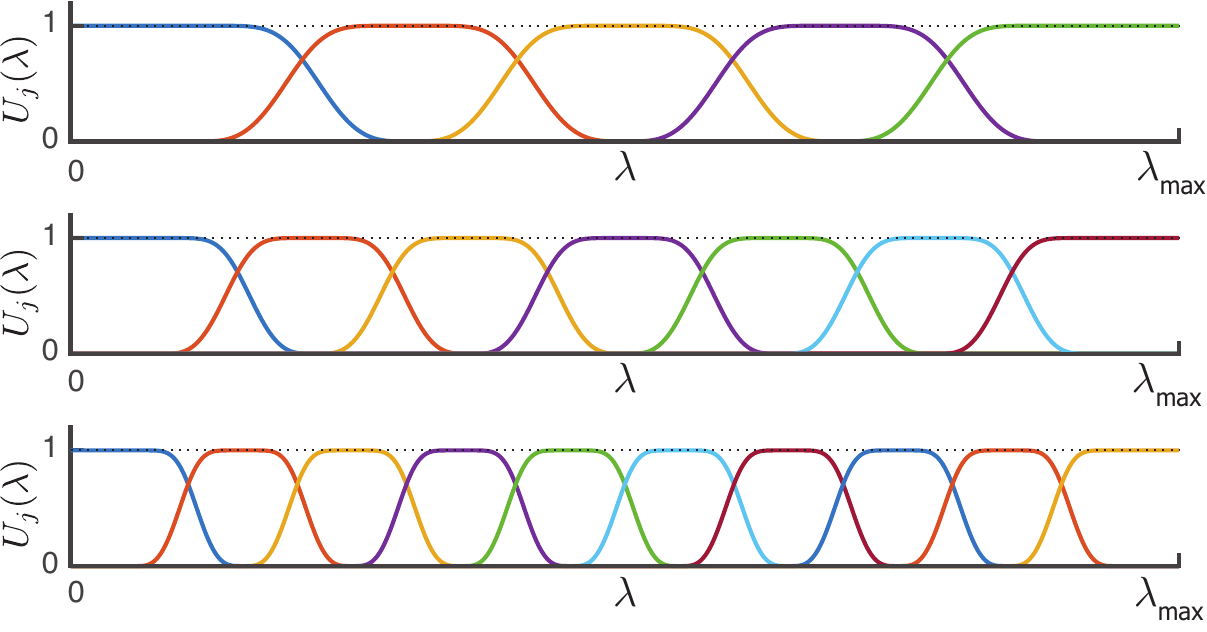}
\caption[]{UMT system of spectral kernels with $J=5$ (top), $J=7$ (middle) and $J=10$ (bottom) spectral scales. The dotted lines correspond to $G(\lambda)$ in (\ref{eq:G2}).} 
\label{fig:umt}
\end{figure}

Figs.~\ref{fig:umt}(a) and (b) show realizations of the resulting UMT system of spectral kernels for a fixed $\lambda_{\textrm{max}}$ and two different $J$. The UMT system of spectral kernels have a narrow passband characteristic with the support of each kernel being a rather strict subset of the spectrum, with minimal overlap of adjacent kernels. 

\section{Ensemble Energy Spectral Density}
\label{sec:ensemble-energy}
The ensemble energy spectral density can be either computed using the graph Fourier transform or approximated through decomposition of the signals using polynomial approximation. In the former approach the ensemble energy is determined at the resolution of eigenvalues whereas in the latter approach it is determined at the resolution of a given number of subbands. The direct computation approach has two shortcomings. Firstly, it requires explicit computation of the graph spectrum and the associated eigenvectors; i.e., a full eigendecompositon of the graph Laplacian matrix,  which is computationally cumbersome for large graphs and infeasible for extensively large graphs. Secondly, it typically results in a non-smooth description of the ensemble energy. These shortcomings are resolved by using the polynomial approximation scheme. 

\subsection{Direct Computation: Using the Graph Fourier Transform}
\subsubsection*{Definition (ensemble energy spectral density)}
For a given graph $G$, with spectrum $\Lambda(G)$, and graph signal set $F = \{\vec{f}_s\}_{s=1}^{N_s}$, the ensemble energy spectral density of $F$ is obtained as
\begin{equation}
\label{eq:ensembleE}
e_{F}[l] = \frac{1}{N_{s}} \sum_{s=1}^{N_{s}} \left | \hat{\widetilde{f}_{s}}[l] \right |^{2}, \quad l = 1,\ldots,N_g,
\end{equation}
where $\widetilde{\vec{f}}_{s}$ denotes the de-meaned and normalized version of $\vec{f}_{s}$ obtained as 
\begin{equation}
\label{eq:g_s}
\widetilde{\vec{f}}_{s} = \frac{\vec{f}_s - \sum_{r=1}^{1+m_{\lambda_{1}}} \langle \vec{f}_s, \vec\chi_{r} \rangle \vec\chi_{r}}{|| \vec{f}_{s} - \sum_{r=1}^{1+m_{\lambda_{1}}} \langle \vec{f}_s, \vec\chi_{r} \rangle \vec\chi_{r}||_{2}}, \quad s=1,\cdots, N_{s}.
\end{equation}
The ensemble energy spectral density has the following properties: (i) $ \{e_{F}[r] = 0\}_{r=1}^{1+m_{\lambda_{1}}}$, and (ii) $ \sum_{l} e_{F}[l] = 1$.

\subsubsection{Approximation: Using Decomposition through Polynomial Approximation}
\label{sec:approx}
The ensemble energy spectral density can be approximated through a multi subband decomposition scheme. Here we use the B-spline based system of spectral kernels, cf.~\ref{sec:bsplineBasedSystem}. The benefit in using a B-spline basis is in the smoothness characteristic of such kernels. Smooth overlapping kernels are advantageous it that i) they enable obtaining a smooth estimation of the ensemble energy spectral density and ii) they can be approximated as low order polynomials. We then decompose the graph signals using the designed system of kernels with a large number of subbands by exploiting the polynomial approximation scheme in decomposition. With such a decomposition, we approximate the ensemble spectral content of the signal set at the resolution of subbands.

Using a system of $N_a$ B-spline based spectral kernels, $\{B_{i}(\lambda)\}_{i=1}^{N_a}$, the ensemble spectral energy of $F$ can be approximated at $N_{a}$ overlapping bands across the spectrum as 
\begin{align}
\label{eq:ensembleE_approx}
a_{F}[i] & = \frac{1}{N_{s}} \sum_{s=1}^{N_{s}} \sum_{n=1}^{N_{g}}  \left | \langle \widetilde{\vec{f}}_{s}, {\vec\psi}_{B_{i},n} \rangle \right |^{2}, \quad i = 1,\ldots,N_a,
\end{align}
where $\widetilde{\vec{f}}_{s}$ is as given in (\ref{eq:g_s}). Let $\vec{b}_{j} \in \ell_{2}(G)$ denote the discrete version of $B_{j}(\lambda)$, i.e., 
\begin{equation}
\label{eq:kernelSampligBeta}
b_{j}[l] = B_{j}(\lambda_{l}), \quad l=1,\ldots, N_{g}.
\end{equation} 

We have $\sum_{i} a_{F}[i] = 1$ since
\begin{align}
\sum_{i} a_{F}[i] 
& \stackrel{(\ref{eq:coeffs3})}{=} \frac{1}{N_{s}} \sum_{i=1}^{N_{a}}\sum_{s=1}^{N_{s}} \sum_{n=1}^{N_{g}}  \left | \sum_{l=1}^{N_g} b_{i}[l] \hat{\widetilde{f}_{s}}[l] \chi_{l}[n] \right |^{2}
\label{eq:bBeta}
\\&\: = \frac{1}{N_{s}} \sum_{s=1}^{N_{s}} \sum_{n=1}^{N_{g}} \Bigg | \sum_{l=1}^{N_g} \underbrace{\sum_{i=1}^{N_{a}} {b_{i}}^{2}[l]}_{\stackrel{(\ref{eq:pouBsplines})}{=}1} \hat{\widetilde{\vec{f}}}_{s}[l] \chi_{l}[n] \Bigg |^{2}
\\&\:= \frac{1}{N_{s}} \sum_{s=1}^{N_{s}} \sum_{n=1}^{N_g}  \left | \sum_{l=1}^{N_{g}}    \hat{\widetilde{f}_{s}}[l]  \chi_{l}[n] \right |^{2}
\\ 
& \stackrel{(\ref{eq:inverseFourier})}{=}  \frac{1}{N_{s}} \sum_{s=1}^{N_{s}} \sum_{n=1}^{N_g} |\widetilde{f}_{s}[n]|^{2}
\\&\: = \frac{1}{N_{s}} \sum_{s=1}^{N_{s}} ||\widetilde{\vec{f}}_{s}||_{2}^{2} 
\\& \stackrel{(\ref{eq:g_s})}{=}1.
\end{align}

If desired, an explicit approximation of the ensemble energy spectral density of $F$, denoted $e_{F}^{(a)}[l]$, can also be determined. First, a continuous ensemble spectral energy representation, denoted $E_{F}^{(a)}(\lambda)$, is obtained through interpolating the set of points 
\begin{equation}
\left \{ (0,0) \cup \left \{ \left ( \frac{\lambda_{\textrm{max}}}{C}  \sum_{k=1}^{i} ||B_{k}(\lambda)||_{2}^{2}\: , \: a_{F}[i] \right ) \right \}_{i=1}^{N_{a}} \right \},
\end{equation}
where $C = \sum_{k=1}^{N_{a}} ||B_{k}(\lambda)||_{2}^{2}$. Then, $e_{F}^{(a)}[l]$ is obtained through sampling $E_{F}^{(a)}(\lambda)$ at $\Lambda(G)$ as 
\begin{equation}
\label{eq:ensembleE_approx}
e_{F}^{(a)}[l] = E_{F}^{(a)} (\lambda_{l}), \quad l=1,\ldots,N_{g}.
\end{equation}

\section{Signal-Adapted System of Spectral Kernels}
\label{sec:signal-adapted}
The construction of a signal-adapted system of spectral kernels is motivated by two observations: (i) the eigenvalues of the graph Laplacian that define the graph's spectrum are irregularly spaced, and depend in a complex way on the graph topology; (ii) the distribution of graph signals' energy is generally non-uniform across the spectrum. Based on these observations, the idea is to construct an \lq adapted\rq~frame, such that the energy-wise significance of the eigenvalues is taken into account, rather than only adapting based on the distribution of the eigenvalues as proposed in \citep{Shuman2015ieee}. In this way, also the topological information of the graph is implicitly incorporated in the design, since the energy content is given in the graph spectral domain that is in turn defined by the eigenvalues. 

For the design of a signal-adapted system of spectral kernels with $J$ subbands, denoted $\{{S}_{j}(\lambda)\}_{j=1}^{J}$, we start off from a prototype system of spectral kernels $\{U_{j}(\lambda)\}_{j=1}^{J}$ that satisfies the following two properties:
\begin{itemize} 
\item (Uniformity constraint)
\begin{equation}
\label{eq:uniformityConstraintContinious}
\exists \: C \in \mathbb{R}^{+}, \quad \int_{0}^{\lambda_{\textrm{max}}} {U}_{j}(\lambda) d\lambda = C, \quad j=1,\ldots,J.
\end{equation}
\item (Tight Parseval frame constraint)
\begin{equation}
\label{eq:tightFrameConstraint2}
\sum_{j=1}^{J} | U_{j}(\lambda)|^{2} = 1, \quad \forall \lambda \in [0,\lambda_{\textrm{max}}].
\end{equation}
\end{itemize}
There is no unique system of kernels that satisfies (\ref{eq:uniformityConstraintContinious}) and (\ref{eq:tightFrameConstraint2}). In this work we use the Meyer-like uniform system of spectral kernels (cf.~\ref{sec:uniformMeyer}) that satisfies the two properties. We then use a suitable spectral-reorganisation transformation to warp the system of kernels to reorganise them along the spectrum to satisfy the desired signal adaptivity as entailed in the spectral-tuning transformation. Here we focus on one specific transform that aims to equalise the amount of ensemble energy that is captured by each kernel.  

If the ensemble spectral density function is available, $T_{F}(\lambda)$ is obtained through monotonic cubic interpolation \citep{monotoniccubicInterpolation} of the pair of points 
\begin{equation}
\label{eq:equiEnergyTransExact}
\left \{ (0,0) 
\cup  \left \{ \left ( \lambda_{l} \: , \:  \alpha_{l} \right ) \right \}_{l=2}^{N_{g}-1}
\cup (\lambda_{\textrm{max}},\lambda_{\textrm{max}})
\\
 \right \},
\end{equation} 
where $Y_l$ is given as
\begin{equation} 
\alpha_{l} = \displaystyle{\frac{ \lambda_{\textrm{max}}} {m_{\lambda_{l}} } \sum_{r=i_{\lambda_{l}}}^{i_{\lambda_{l}}+m_{\lambda_{l}}} \sum_{k=1}^{r} e_{F}[k]}.
\end{equation} 
If the ensemble energy spectral density is approximated using a system of $N_{a}$ B-spline based spectral kernels, cf.~(\ref{sec:approx}), $T_{F}(\lambda)$ can instead be obtained through monotonic cubic interpolation of the set of points 

\begin{equation}
\label{eq:equiEnergyTransApprox}
\left \{ (0,0) \cup \left \{ \left (\omega_{l}, \alpha_{l} \right ) \right \}_{i=1}^{N_{a}-1} \cup (\lambda_\textrm{max},\lambda_\textrm{max}) \right \},
\end{equation}
where $C  = \sum_{k=1}^{N_{a}} ||B_{k}(\lambda)||_{2}^{2}$, and pairs $(\omega_{l}, \alpha_{l})$ are given as
\begin{align}
\omega_{l} & = \frac{\lambda_{\textrm{max}}}{C}  \sum_{k=1}^{i} ||B_{k}(\lambda)||_{2}^{2}
\\
\alpha_{l} & = \lambda_{\textrm{max}} \sum_{k=1}^{i} a_{F}[k].
\end{align}
By incorporating a desired $T_{F}(\lambda)$ in $\{U_{j}(\lambda)\}_{j=1}^{J}$, a warped version of the prototype design is obtained as
\begin{align}
\label{eq:warpedFrame}
 S_{j}(\lambda) = U_{j}(T_{F}(\lambda)), \quad j=1, \ldots, J.
\end{align}
We refer to $\{S_{j}(\lambda)\}_{j=1}^{J}$ as a signal-adapted system of spectral kernels. The atoms of a dictionary constructed using $\{S_{j}(\lambda)\}_{j=1}^{J}$ form a Parseval frame on $\ell_{2}(G)$ since 
\begin{align}
\sum_{j=1}^{J} \vert S_{j}(\lambda)\vert ^{2} 
& \stackrel{(\ref{eq:warpedFrame})}{=}  \sum_{j=1}^{J} \vert U_{j}(\underbrace{T_{F}(\lambda)}_{:=\lambda^{'}})\vert ^{2}, \quad \forall \lambda \in  [0,\lambda_{\text{max}}] 
\nonumber
\\
& \: \: = \sum_{j=1}^{J} \vert U_{j}(\lambda^{'})\vert ^{2}, \quad \forall \lambda^{'} \in  [0,\lambda_{\text{max}}]
\nonumber
\\
& \: \: = 1
\nonumber
\end{align}
where the last equality follows from Proposition~\ref{propo:umt}. For direct decomposition as in (\ref{eq:coeffs3}), a discrete representation $\{{\boldsymbol{s}}_j\}_{j=1}^{J}$ can be obtained through sampling $S_{j}(\lambda)$ at $\Lambda(G)$. With this design, each of the $J$ spectral kernel $\{{\boldsymbol{s}}_j\}_{j=1}^{J}$ capture an equal amount of \textit{ensemble} energy. That is, if the ensemble energy spectral density is used we have   
\begin{equation}
\label{eq:equiEnergyConstraint}
\sum_{l=1}^{N_{g}} {\boldsymbol{s}}_{j}[l] \boldsymbol{e}_{F}[l] \stackrel{(\ref{eq:ensembleE})}{=} \frac{1 }{J}, \quad j=1,\ldots,J,   
\end{equation}
and if the approximation scheme is used we have 
\begin{equation}
\label{eq:equiEnergyConstraintApprox}
\sum_{l=1}^{N_{g}} {\boldsymbol{s}}_{j}[l] \boldsymbol{e}_{F}^{(a)}[l] \stackrel{(\ref{eq:ensembleE_approx})}{=} \frac{1 }{J}, \quad j=1,\ldots,J.  
\end{equation}

\section{Example Designs and Applications of Signal-Adapted Decomposition of Graph Signals}
We present constructions of signal-adapted systems of spectral kernels for signal sets realized on the Minnesota road graph, as well as three different brain graphs: template-based cerebellum gray matter graph \citep{Behjat2015}, individualised cerebral cortex graphs~\citep{Milloz2024biorxiv}, and individualised white matter graphs~\citep{Abramian2021}.

\subsection{The Minnesota Road Graph}
\label{sec:minnesota}
The edges of the Minnesota Road Graph represent major roads and its vertices their intersection points, which often correspond to towns or cities, see Fig.~\ref{fig:minnesotaGraph}(a). Fig.~\ref{fig:minnesotaGraph}(b) shows the graph's normalized Laplacian spectrum presented as the distribution of the eigenvalues.  

Given the absence of real data for this commonly used benchmark graph, we consider a model for simulating random graph signals of varying smoothness. For a given graph with adjacency matrix $A$, a general model for realizing graph signals of density $\eta \in ]0,1]$ and smoothness $n \in \mathbb{Z}^{+}$ ca be defined as   
\begin{equation}
\label{eq:signalModel}
\vec{x}_{\eta,n} = {A}^{n} \vec{p}_{\eta},
\end{equation}
where $\vec{p}_{\eta} \in \ell_{2}(G)$ denotes a random realization of a spike signal as $\{\vec{p}_{\eta}[i] \in \{0,1\}\}_{i=1,\ldots,N_g}$ such that $\sum_i \vec{p}_{\eta}[i] =  \eta N_g$. Application of the $n$-th power of $A$ to $\vec{p}_{\eta}$ leads to a signal that i) respects the intrinsic structure of the graph and ii) has a desired smoothness determined by $n$, a higher $n$ leading to a smoother graph signal. 

\begin{figure}[]
\centering
\includegraphics[width=0.48\textwidth]{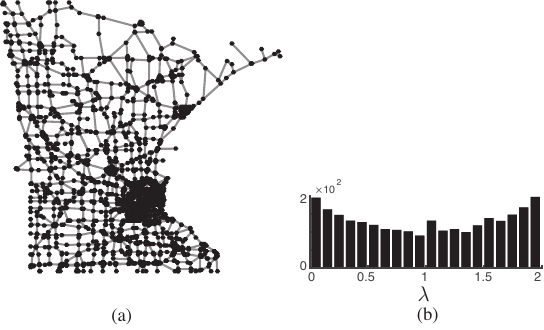}
\caption[]{(a) Minnesota road graph. (b) Histograms of the eigenvalues $\Lambda_{\mathcal{L}}(G)$ of the Minnesota road graph. Each bar indicates the number of eigenvalues that lie in the corresponding spectral range.}
\label{fig:minnesotaGraph}
\end{figure}

\begin{figure}[]
\centering
\includegraphics[width=0.48\textwidth]{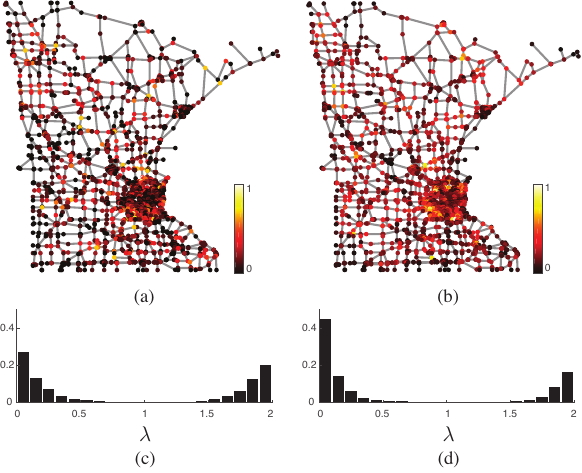}
\caption[]{Sample signal realizations on the Minnesota road graph, (a) $\vec{x}_{0.2,2}$ and (b) $\vec{x}_{0.5,4}$. The plots are normalized as $\vec{x}_{\eta,n}/||\vec{x}_{\eta,n}||_{\infty}$ (c)-(d) Distribution of the ensemble energy spectral density $\boldsymbol{e}_{F_{1}}$ and $\boldsymbol{e}_{F_{2}}$, respectively. Each bar indicates the sum of ensemble energies of the eigenvalues lying in the corresponding spectral range.}
\label{fig:minnesotaSignals}
\end{figure}

\begin{figure*}[]
\centering
\includegraphics[width=0.6\textwidth]{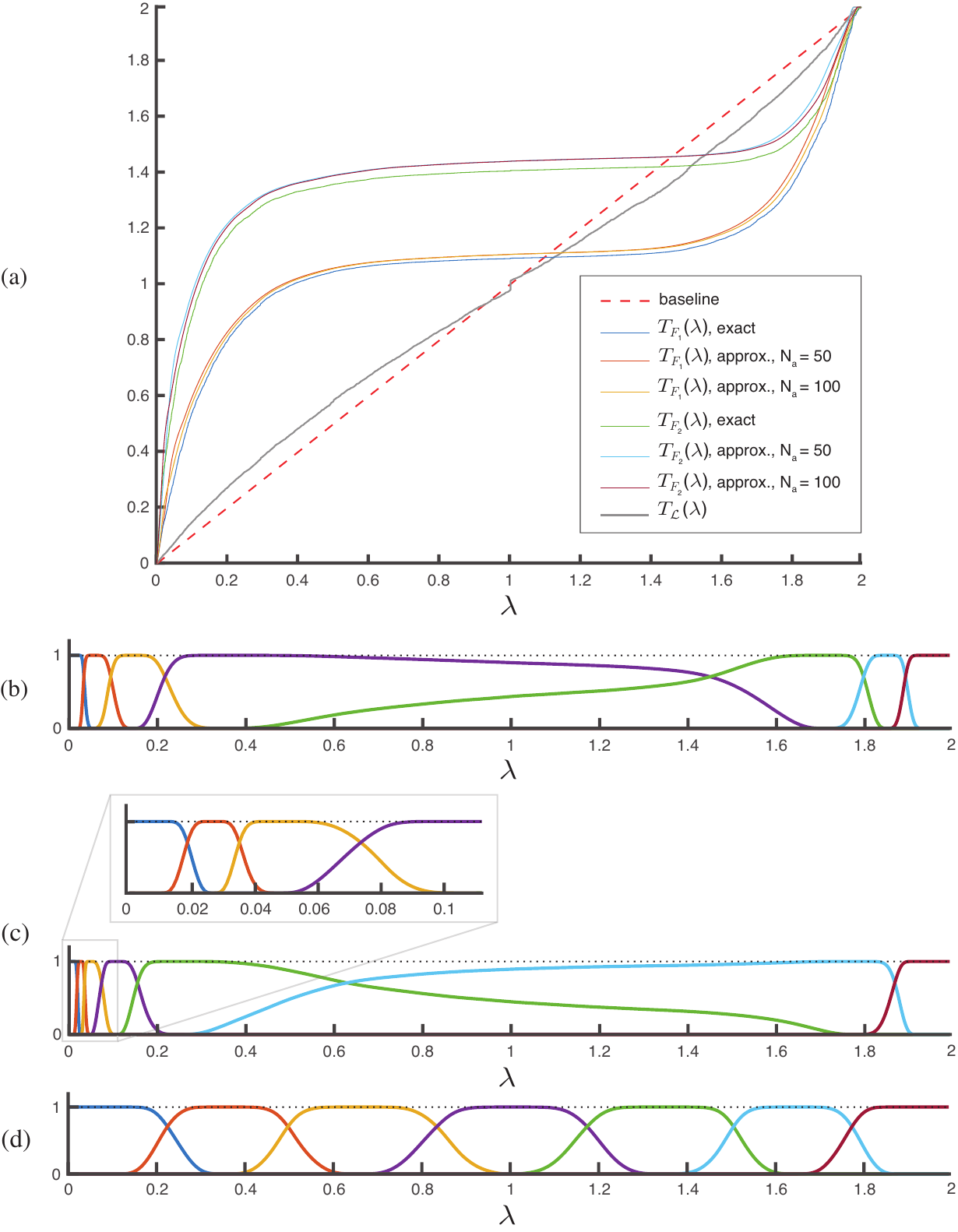}
\caption[]{(a) Constructed energy-equalizing transformation functions, $T_{F_{1}}(\lambda)$ and $T_{F_{2}}(\lambda)$ using the exact and approximation schemes. $N_{a}$ denotes the number of spectral kernels used for the approximation, cf. (\ref{eq:ensembleE_approx}). (b)-(c) Signal-adapted system of spectral kernels constructed by warping the UMT system of spectral kernels ($J=7$) using $T_{F_{1}}(\lambda)$ (approx, $N_{a}=100$) and $T_{F_{2}}(\lambda)$ (approx., $N_{a}=100$), respectively. (d) Spectrum-adapted system of spectral kernels constructed by warping the UMT system of spectral kernels ($J=6$) using $T_{\mathcal{L}}(\lambda)$. In (b)-(d), the dotted lines corresponds to $G(\lambda)$ in (\ref{eq:G2}).}
\label{fig:minnesotaResults}
\end{figure*}

Two sets of graph signals were constructed as 
\begin{align*}
F_{1} =  \left \{ \left \{ \vec{x}_{\eta,2}^{[i]} \right \}_{\eta =0.2,0.5}\right\}_{i=1,\ldots,10},
\\
F_{2}  =  \left \{ \left \{ \vec{x}_{\eta,4}^{[i]} \right \}_{\eta =0.2,0.5}\right\}_{i=1,\ldots,10},
\end{align*}
where index $i$ denotes random realizations of $\boldsymbol{p}_{\eta}$ in (\ref{eq:signalModel}), resulting in 20 signals in each set. Figs.~\ref{fig:minnesotaSignals}(a) and (b) show a realizations of a signal from $F_{1}$ and $F_{2}$, respectively. 

Fig.~\ref{fig:minnesotaResults}(a) shows the energy-equalizing transformation functions associated to $F_{1}$ and $F_{2}$. The transformations constructed based on $\boldsymbol{a}_{F_{\cdot}}$, cf. (\ref{eq:equiEnergyTransApprox}) closely matches that constructed based on $\boldsymbol{e}_{F_{\cdot}}$, cf. (\ref{eq:equiEnergyTransExact}). The former transformation has the benefit of being smooth, and indeed, that it was computed without the explicit need to diagonalize $L$. By incorporating the transformations in the UMT system of spectral kernels, signal-adapted systems of spectral kernels are obtained, see Figs.~\ref{fig:minnesotaResults}(b)-(c). 

\begin{figure*}[]
   \centering
   \includegraphics[width=0.7\textwidth]{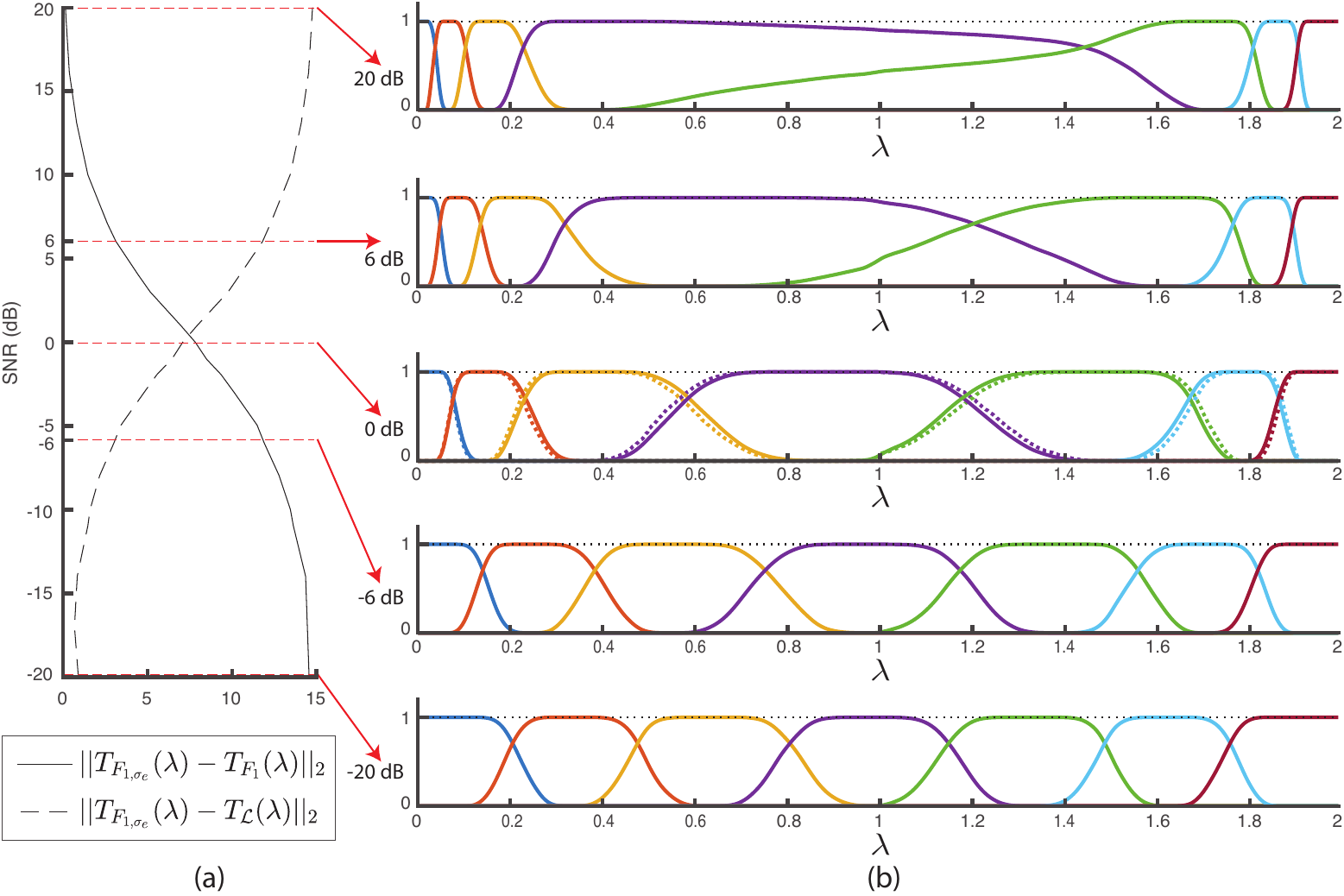} 
   \caption{(a) Deviation of energy-equalizing transformation functions of noise added signal sets $T_{F_{1,\sigma_{e}}}(\lambda)$ relative to $T_{\mathcal{L}}(\lambda)$ and $T_{F_{1}}(\lambda)$ (cf. Fig.~\ref{fig:minnesotaResults}(a)) as a function of the signal sets' SNRs. (b) Signal-adapted system of spectral kernels constructed by warping the UMT system of spectral kernels ($J=7$) using $T_{F_{1,\sigma_{e}}}(\lambda)$ of noise-added signal sets at five different SNRs. At 0 dB, the resulting system of kernels are overlaid on the system of kernels obtained by warping the UMT system of spectral kernels using the transformation function  $(T_{F_{1}}(\lambda) + T_{\mathcal{L}}(\lambda) )/2$, shown in dashed lines.}
   \label{fig:robustnessToNoise}
\end{figure*}

A comparison of Figs.~\ref{fig:minnesotaResults}(b) and (c) and Figs.~\ref{fig:minnesotaSignals}(c) and (d) highlights the energy-wise optimality of the proposed signal-adapted frame construction; i.e., more filters are allocated to spectral ranges  that have higher ensemble energy. The support of the filters in the two sets vary relative to the difference in the distribution of the ensemble energy of the two signal sets, with more filters allocated to the lower end of the spectrum for the $F_{2}$ frame than for the $F_{1}$ frame, and vice versa at the upper end of the spectrum. For comparison, a spectrum-adapted system of kernels is shown in Figs.~\ref{fig:minnesotaResults}(d). The spectrum-adapted system of kernels is obtained by warping the UMT prototype system of kernels with a spectrum-equalizing transformation function $T_{\mathcal{L}}(\lambda)$ which equalizes the distribution of the eigenvalues \citep{Shuman2015ieee}. As the distribution of the eigenvalues of the Minnesota Road graph minimally deviate from a uniform distribution, so does the spectrum-adapted system of kernels relative to the UMT prototype, compare Figs. \ref{fig:umt} and \ref{fig:minnesotaResults}(d). On the contrary, the signal-adapted design optimizes the construction of the kernels such that the energy-wise significance of the eigenvalues is taken into account, rather than only considering the distribution of the eigenvalues as in the spectrum-adapted frame. Such adaptation results in a system of spectral kernels that largely deviate from the UMT prototype.   

\subsubsection{Robustness to Noise}
It is interesting to study the robustness of the design to possible additive noise. Let $F_{1,\sigma_{e}}$ denote the noise added version of signal set $F_{1}$ computed as 

\begin{equation}
F_{1,\sigma_{e}} =  \left \{ \vec{y}_{i} = \vec{x}_{i} + \vec{e}_{i}~|~\vec{x}_{i} \in F_{1} \right\}_{i=1,\ldots,20}, 
\end{equation}
where $\{\vec{e}_{i}\}_{i=1}^{20}$ denote random realizations of additive white Gaussian noise of standard deviation $\sigma_{e}$. We construct signal sets $F_{1,\sigma_{e}}$ of varying SNR $= \sigma_{x}^{2} / \sigma_{e}^{2}$, where $\sigma_{x}$ denotes the standard deviation of each signal $\vec{x}_{i} \in F_{1}$. Let $T_{F_{1,\sigma_{e}}}(\lambda)$ denote the energy-equalizing transformation function associated to $F_{1,\sigma_{e}}$. Fig.~\ref{fig:robustnessToNoise}(a) shows  mean-square error metrics $|| T_{F_{1,\sigma_{e}}}(\lambda) - T_{F_{1}}(\lambda) ||_{2}$ and $|| T_{F_{1,\sigma_{e}}}(\lambda) - T_{\mathcal{L}}(\lambda) ||_{2}$ across signal sets $F_{1,\sigma_{e}}$ of varying SNR, where $T_{\mathcal{L}}(\lambda)$ and $T_{F_{1}}(\lambda)$ are the transformation functions shown in Fig.~\ref{fig:minnesotaResults}(a), $T_{F_{1}}(\lambda)$ being the approximated version using $N_{a}=100$. The estimated energy-equalizing transformation functions $T_{F_{1,\sigma_{e}}}(\lambda)$ become more similar to $T_{F_{1}}(\lambda)$ as the SNR increases. At low SNRs, $T_{F_{1,\sigma_{e}}}(\lambda)$ become more similar to $T_{\mathcal{L}}(\lambda)$. The signal-adapted system of spectral kernels using noise-added signal sets of five different SNRs are shown in Fig.~\ref{fig:robustnessToNoise}(b). At the two extremes, i.e., +20 dB and -20dB, the system of kernels become almost identical to the system of kernels shown in Figs.~\ref{fig:minnesotaResults}(b) and (d), respectively. At 0dB, the signal-adapted system of kernels at each subband can be seen as the average of the corresponding kernels in the associated subbands in Figs.~\ref{fig:minnesotaResults}(b) and (d). Equivalently, this can be seen as constructing a system of kernels through warping the the UMT prototype system of kernels with a warping function defined as the average of the spectrum-equalizing and energy-equalizing transformation functions, i.e., $(T_{F_{1}}(\lambda) + T_{\mathcal{L}}(\lambda) )/2$, see Fig.~\ref{fig:robustnessToNoise}(b) at 0 dB.        

\subsection{Multiscale Characterization of Brain Cortical Maps}
GSP is highly effective in brain neuroimaging, with the objective to study and characterize brain anatomy, function, pathology, and their interactions. A common approach constructs a brain graph from anatomical features such as cortical morphology (from structural MRI) or white matter fiber architecture (from diffusion MRI). The brain's gray matter forms a convoluted structure interleaved with white matter and cerebrospinal fluid, while white matter consists of axonal fiber pathways. These structures enable detailed spatial modeling of individual brain anatomy as a graph.

Functional (e.g., functional MRI) and pathological (e.g., PET) neuroimaging data can be treated as graph signals, representing functions on a brain graph. Functional MRI (fMRI), a standard technique for studying brain function, captures blood-oxygen-level-dependent (BOLD) signals, which indicate changes in blood flow to active regions. Historically analyzed in gray matter \citep{Logothetis2004}, recent studies have demonstrated its presence also in white matter \citep{Li2019, Abramian2021, Schilling2023, Huang2023, Zu2024}. In both tissues, the BOLD signal exhibits spatial patterns that are not well suited for analysis in a classical Euclidean framework, where filters and wavelets are typically isotropic and shift-invariant. Signal decompositions using Euclidean-based designs are thus inadequate for detecting non-isotropic BOLD signal in gray matter as well as elongated, low-amplitude BOLD signal that aligns with fiber bundles and crosses gray matter boundaries. The same limitations hold for PET imaging data, specifically, PET-based manifestation of pathological proteins such as tau and amyloid-beta in Alzheimer's disease, which manifest notably heterogeneity in their spatial profiles that is linked to underlying anatomical boundaries.    

Anatomically adapted graph wavelets \citep{Behjat2015} and low-pass filters tailored for white matter \citep{Abramian2021, Behjat2025} have been proposed to address the need for designs that adapt to the domain of studied signals. However, challenges remain in designing flexible frames that suitably partition the spectrum for analysis of fMRI~\citep{Behjat2020isbi, Behjat2021embc, Ferritto2023}. This motivates the need for a frame design that adapts to the spectral characteristics of fMRI and PET graph signals. Here we examine brain graph signal decompositions optimized for the energy distribution of fMRI and PET data, treated as graph signals on three types of anatomical brain graphs. In the first two applications (Sections~\ref{sec:cerebellum} and \ref{sec:asymmetry}), we use the design methodology introduced in Section~\ref{sec:signal-adapted} while in the third example (Section~\ref{sec:cerebellum}) we show how a signal-adapted decomposition can be implemented more flexibly with user defined criteria.  

\subsubsection{Cerebellum graph and fMRI graph signals}
\label{sec:cerebellum}
We construct a graph encoding the 3D topology of cerebellar gray matter~\citep{Behjat2015} using an atlas-based template~\citep{suitAtlas}. Graph vertices represent gray matter voxels, while edges are defined by $3\times3\times3$ voxel neighborhood adjacency (Fig.~\ref{fig:cerebellumGraph}). fMRI data were collected from 26 healthy subjects performing an event-related visual stimulation task~\citep{Kelly}. Each subject underwent structural MRI and functional scans, which were co-registered and resampled to ensure voxel-wise alignment. Functional voxels corresponding to cerebellar gray matter were then extracted and treated as cerebellar graph signals. For each subject, a signal set $\{F_{k}\}_{k=1}^{26}$ was created by randomly selecting 20 functional signals. A combined signal set, $F = F_{1} \cup F_{2} \cup \cdots \cup F_{26}$, was also constructed across all subjects. 

\begin{figure}[b]
\centering
\includegraphics[width=.48\textwidth]{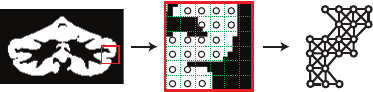} 
\caption[]{Illustration of the cerebellum graph.}
\label{fig:cerebellumGraph}
\end{figure}

\begin{figure}[h]
\centering
\includegraphics[width=.48\textwidth]{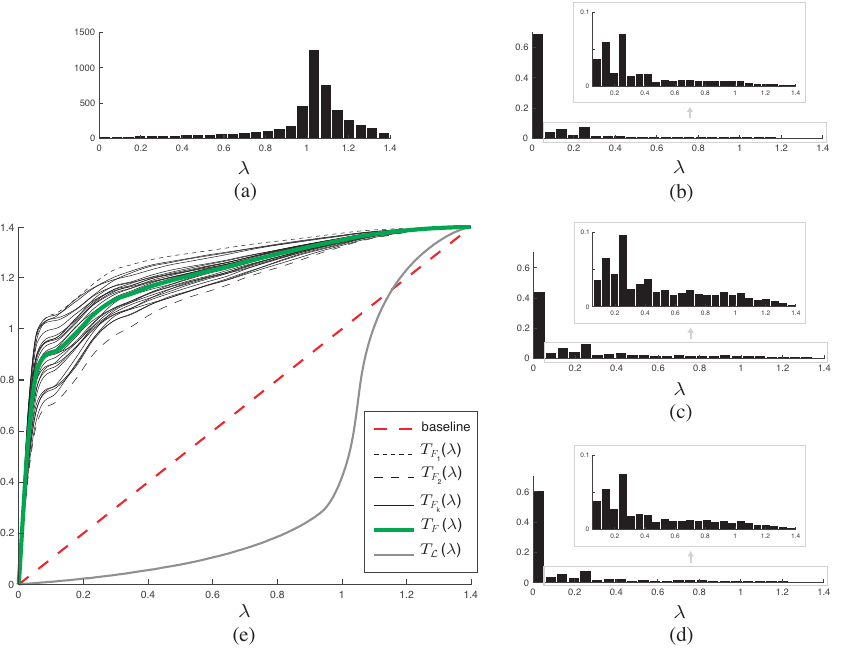}
\caption[]{(a) Histogram of the eigenvalues $\Lambda_{\mathcal{L}}(G)$ of the cerebellum graph. 
(b)-(d) Distribution of the ensemble energy spectral density of $F_{1}$, $F_{2}$ and $F$. 
(c) Energy-equalizing and spectrum-equalizing transformation functions. The black curves correspond to the energy-equalizing transformation for each subject's signal set. The upper and lower extreme transformations represented with dashed curves are associated to signal sets $F_{1}$ and $F_{2}$, respectively.}
\label{fig:cerebellum-p1}
\end{figure}

\begin{figure}[]
\centering
\includegraphics[width=.48\textwidth]{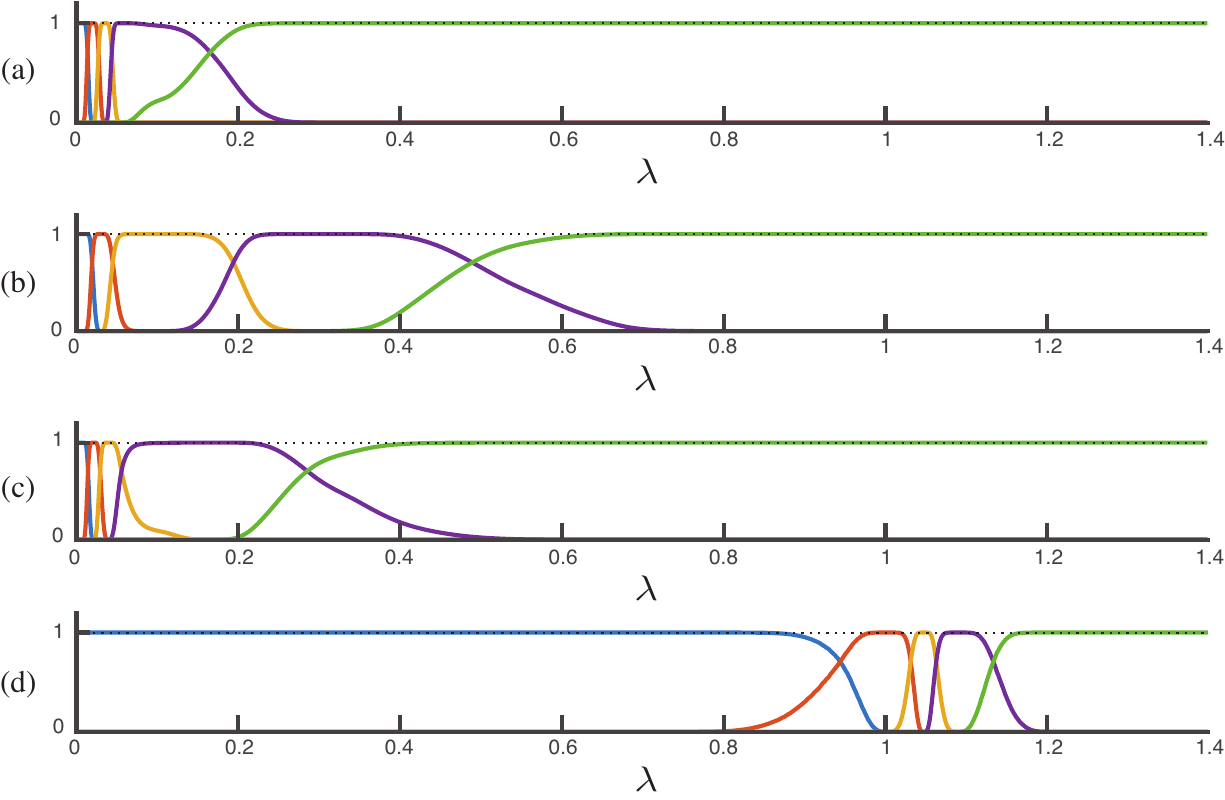}
\caption[]{
(a)-(c) Signal-adapted system of spectral kernels adapted to the ensemble spectral content of cerebellum graph signal sets $F_{1}$, $F_{2}$ and $F$, respectively. 
(d) Spectrum-adapted system of spectral kernels.}
\label{fig:cerebellum-p2}
\end{figure} 
  
Fig.~\ref{fig:cerebellum-p1}(a) shows the distribution of the eigenvalues of the cerebellum graph. The distribution of the ensemble energy spectral density of signals sets $F_{1}$, $F_{2}$ and $F$ are shown in Figs.~\ref{fig:cerebellum-p1}(b), (c) and (d), respectively. The distribution of eigenvalues is substantially different from that of the ensemble energy spectral densities; most eigenvalues are located at the upper end of the spectrum, whereas the ensemble energy is largely concentrated at the lower end of the spectrum. The ensemble energy spectral densities also vary across the signal sets. Signal set $F_{1}$ has more low energy spectral content than $F_{2}$ (compare the height of the first bins of the histograms in Figs.~\ref{fig:cerebellum-p1}(b) and (c)), whereas $F_{2}$ shows greater spectral content at higher harmonics. $F_{1}$ and $F_{2}$ represent the two extremes across the subjects. The distribution of the ensemble energy content of $F$ falls in between that of $F_{1}$ and $F_{2}$, see Figs.~\ref{fig:cerebellum-p1}(d). This is better observed by comparing the energy-equalizing transformation functions, see Fig.~\ref{fig:cerebellum-p1}(e). The transformations associated to $\{F_{k}\}_{k=3}^{26}$ span the space in between $T_{F_{1}}(\lambda)$ and $T_{F_{2}}(\lambda)$, and $T_{F}(\lambda)$ falls almost in the mid range. 

\begin{figure*}[h]
\centering
\includegraphics[width=.65\textwidth]{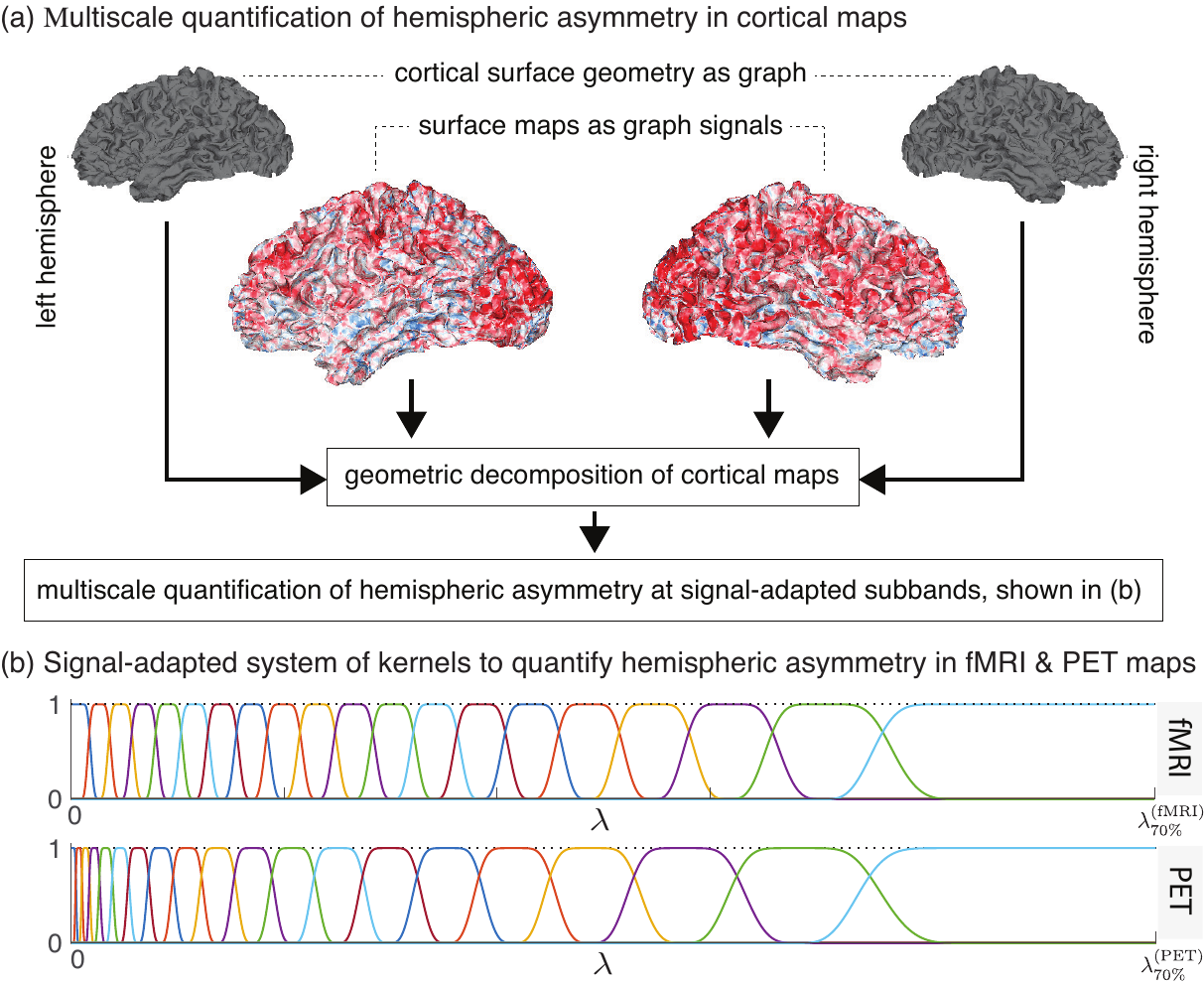}
\caption[]{Signal-adapted, multiscale quantification of hemispheric asymmetry in brain cortical maps. (a) Schematic overview of the methodology. (b) Signal-adapted frames covering the lower-end of the spectrum up to the spectral value that on average 70\% of the total ensemble spectral energy falls for fMRI (top) and tau-PET (bottom) cortical maps. Figure content based on \citep{Milloz2024biorxiv}.}
\label{fig:asymmetry-p1}
\end{figure*}

The stark contrast between eigenvalue and ensemble energy distributions results in a major discrepancy between? $T_{\mathcal{L}}(\lambda)$ and the energy-equalizing transformations. Fig.~\ref{fig:cerebellum-p2} illustrates the resulting signal-adapted and spectrum-adapted systems of spectral kernels. The spectrum-adapted frame's kernels are concentrated at the higher end of the spectrum, where many eigenvalues lie, while signal-adapted frame kernels are localized at the lower end. This suggests that the signal-adapting scheme optimally configures filters based on ensemble energy content. The proposed signal-adapted frame's narrowband configuration at the lower spectrum end closely aligns with prior cerebellar data analysis methods \citep{Behjat2015}, which were empirically tuned using a Meyer-like graph wavelet frame \citep{Leonardi2013}.

\subsubsection{Cortical surface graph and fMRI/PET graph signals}
\label{sec:asymmetry}

Hemispheric asymmetry is a fundamental feature of brain organization with implications for function, structure, and disease. Traditional methods for measuring hemispheric asymmetry rely on aggregate brain morphology metrics, such as surface area and gray matter volume, which capture only single spatial scales. For this type of analyses, the cortex is divided into atlas-defined regions, and laterality is determined by subtracting corresponding values in the left and right hemispheres. Cortical maps, however, exhibit both long- and short-range spatial frequency characteristics. Furthermore, the cerebral cortex is a continuous, highly convoluted structure that cannot be seamlessly partitioned into discrete regions at the individual level using non-invasive imaging.

Increasing evidence suggests that brain eigenmodes offer a powerful, generalized framework for studying the cortex's multiscale structural~\citep{Cao2024, Maghsadhagh2021,  Mansour2025, Wachinger2015} and functional organization~\citep{Behjat2025, HuangBolton2018, Mansour2024, Miri2024, Olsen2024, Tarun2020}. A harmonic framework based on cortical surface eigenmodes can comprehensively capture spatial patterns and quantify lateralisations. Here, we apply the signal-adapted frame design methodology to define spatial subbands optimized for capturing lateralized spatial content in cortical fMRI and PET maps~\citep{Milloz2024biorxiv} (Fig.~\ref{fig:asymmetry-p1}(a)).

We study fMRI maps---from healthy young adults~\citep{VanEssen2013}---that represent each region's functional connectivity, characterizing functional correlations/anti-correlations across the cortex. The PET maps---from cognitively unimpaired (CU) individuals and patients with Alzheimer's disease (AD) dementia~\citep{Palmqvist2020}---quantify tau pathology, a key marker of AD dementia. We focus on the lower-end of the spectrum, capturing up to 70\% of total signal energy in the studied maps, which corresponds to studying approximately the first 2000 eigenmodes when treating fMRI maps as graph signals and the first 500 when tau-PET maps are studied. This difference reflects the greater presence of larger spatial frequency content in fMRI connectivity maps compared to tau-PET maps. Fig.\ref{fig:asymmetry-p1}(b) shows the resulting signal-adapted system of kernels, revealing subtle differences in subband configurations between the two datasets.

\begin{figure*}[]
\centering
\includegraphics[width=.65\textwidth]{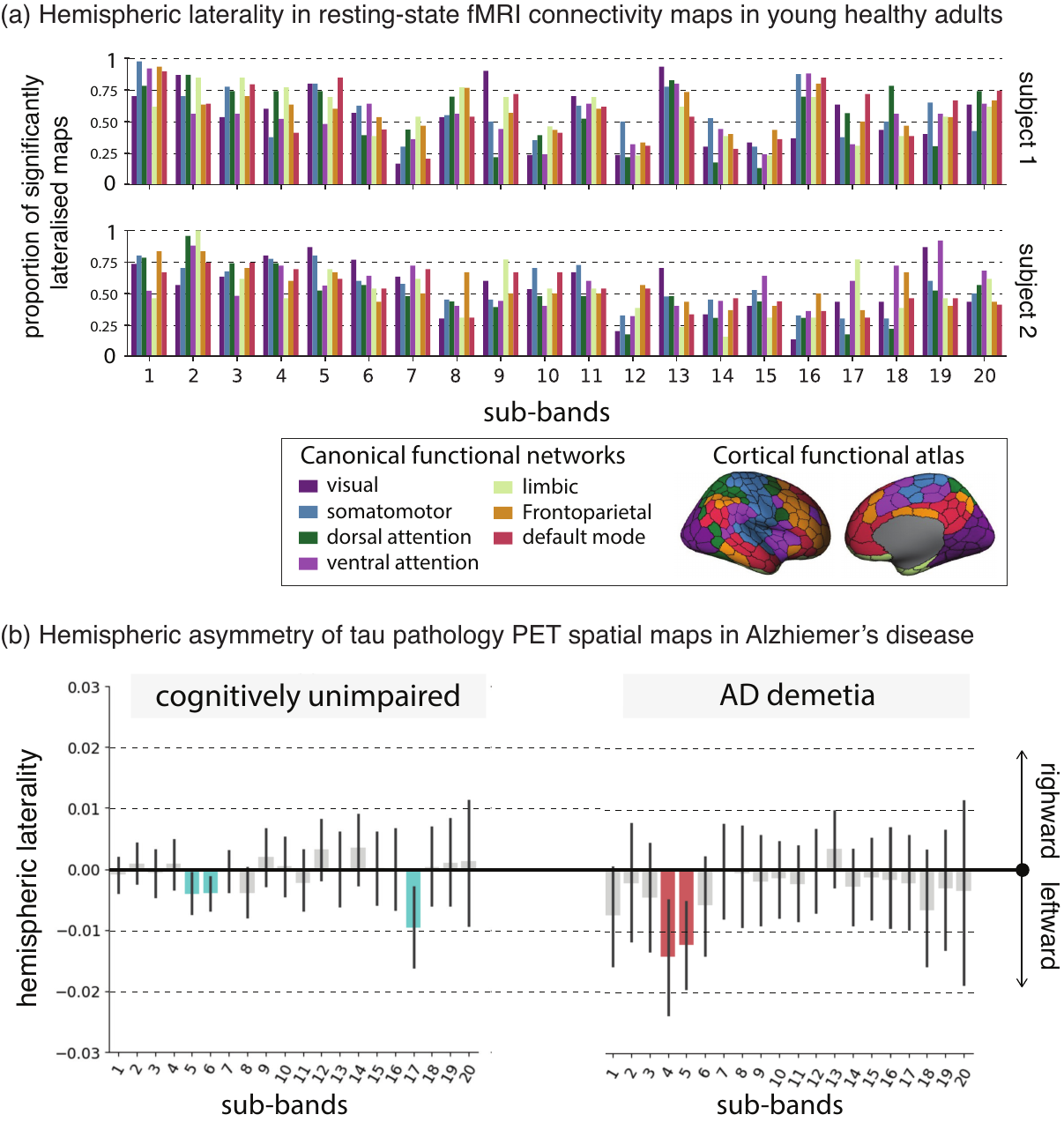}
\caption[]{Signal-adapted quantification of hemispheric asymmetry in brain cortical maps. (a) Proportion of fMRI seed connectivity maps that manifest significant laterality in the different spectral bands. Results shown for two representative subjects. (d) Hemispheric asymmetry in tau-PET maps in cognitively unimpaired patients and in patients with AD dementia. Figure content based on \citep{Milloz2024biorxiv}.}
\label{fig:asymmetry-p2}
\end{figure*}

For the fMRI data, cortical regions in different canonical functional networks have connectivity patterns that manifest different degrees of hemispheric asymmetry in the different subbands (Fig.~\ref{fig:asymmetry-p2}(a)). The results also reveal subject-specific idiosyncrasies in asymmetry across the different subbands and canonical networks. The greatest asymmetry is observed in the frontal, posterior parietal, and lateral temporal cortices. 

For the tau-PET data, aggregates of the pathological tau protein manifest subtle asymmetries at varying spatial scales, which notably differs between CU individuals and patients with AD dementia (Fig.~\ref{fig:asymmetry-p2}(b)). In both groups, significant leftward asymmetry is observed in several subbands (highlighed subbands), with the laterality being substantially larger in the patients with AD dementia in the geometric spatial scale that is represented by the fourth and fifth signal-adapted subbands. 
 
\subsubsection{White matter graph and fMRI graph signals}
To demonstrate the flexibility of signal-adapted graph signal decompositions, in this final application we employ a decomposition without the explicit use of the ``warping procedure'' described under Section~\ref{sec:signal-adapted}, i.e., Eq.~(\ref{eq:warpedFrame}). However, it still utilizes ensemble spectral energy estimates to guide the construction of spectral kernels tailored to the spectral characteristics of the signals at hand.
  
We construct individualized graphs encoding white matter fiber architecture using diffusion MRI data~\citep{Abramian2021}, as shown in Fig.\ref{fig:whitematter}(a). For each subject (95 total) and hemisphere, a graph is built where white matter voxels serve as vertices. Edges between adjacent voxels are weighted based on the coherence of their diffusion orientation distribution functions (ODFs)---higher weights indicate better alignment with the connecting edge. We refer the interested reader to~\citep{Abramian2021} for further design details.

\begin{figure*}[]
\centering
\includegraphics[width=.65\textwidth]{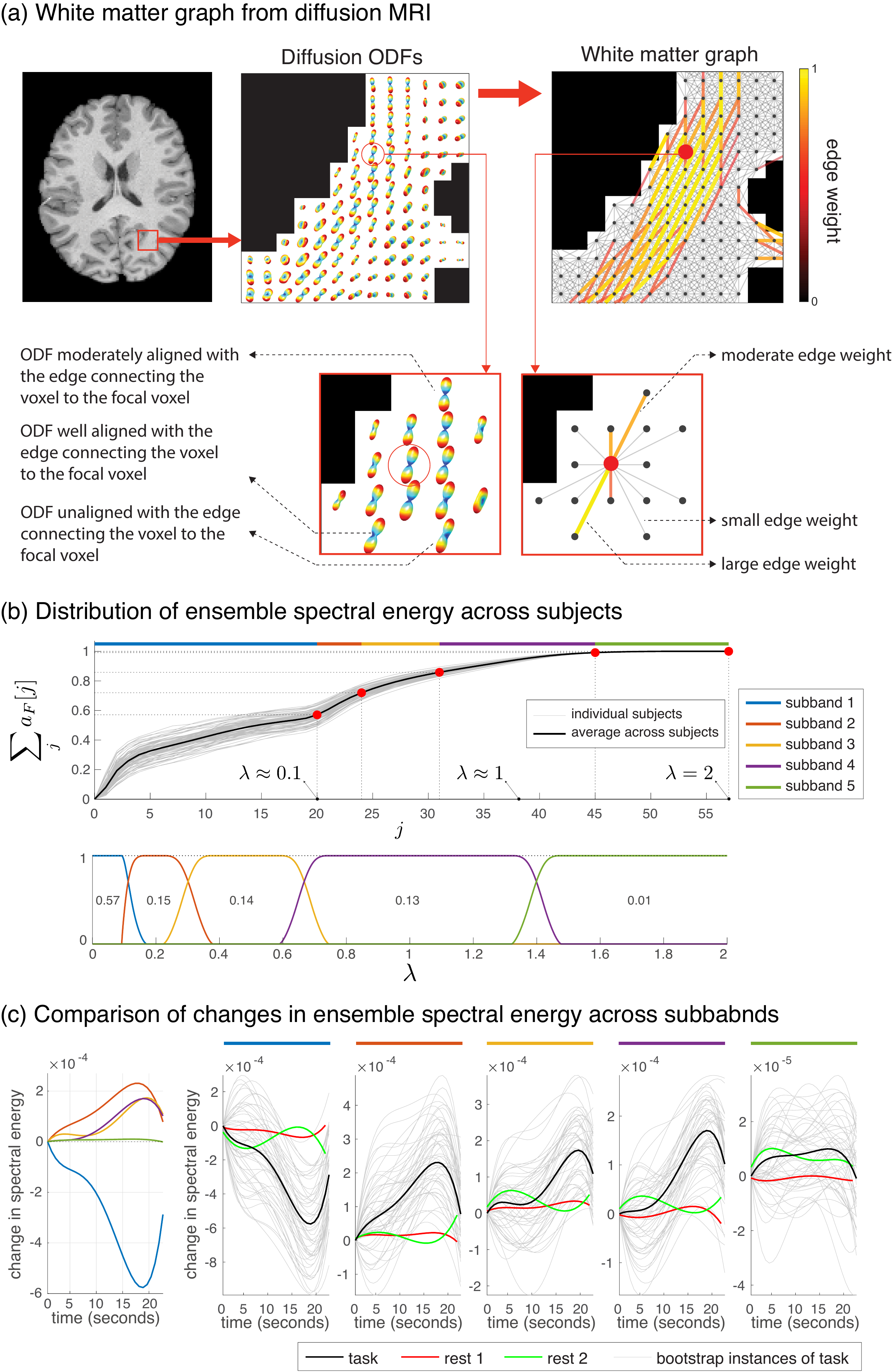}
\caption[]{Signal-adapted decomposition of fMRI data in white matter. (a) A schematic illustration of the white matter graph of a representative subject. (b) Ensemble spectral energy estimates using the system of kernels shown in Fig.~\ref{fig:sosks57} across 95 subjects. A five-band system of spectral kernels is defined based on these estimates. The first kernel covers the lower 5\% of the spectrum, the last spans the region capturing 1\% of total signal energy, and the remainder of the spectrum is divided into three bands, each capturing roughly equal amount of total ensemble signal energy; values shown in the kernels are fraction of total ensemble captured per band. (c) Figure content based on \citep{Abramian2021, Behjat2021isbi}.}
\label{fig:whitematter}
\end{figure*}

Fig.\ref{fig:whitematter}(b) illustrates the spectral energy distribution of white matter fMRI graph signals across subjects, estimated using the multi-resolution spectral kernel system (Fig.\ref{fig:sosks57}) via polynomial approximation (cf. Section~\ref{sec:approx}). Over 50\% of total spectral energy is concentrated in eigenvectors with eigenvalues in $[$0, 0.1$]$, while less than 1\% is in $[$1.4, 2$]$. Overall, more than 90\% of ensemble signal energy is captured by eigenvectors in the lower half of the spectrum. Given this distribution of energies, and in order to reduce the dimensionality of the spectral representation, we studied variations in the spectral energy of white matter fMRI maps using a coarser set of five spectral kernels as shown at the bottom of Fig.\ref{fig:whitematter}(b). This system of kernels, which form a Parseval frame, was defined by merging the narrowband kernels of the system of kernels (i.e. Fig.\ref{fig:sosks57}) used to estimate the energy content: the first kernel covers the lower 5\% of the spectrum (first 20 kernels of Fig.\ref{fig:sosks57}, as in prior studies~\citep{Behjat2020isbi}), the last spans the upper-end of the spectrum the eigenmode of which capture 1\% of total signal energy, and the remainder of the spectrum is divided into three bands, each capturing roughly an equal amount of the total spectral energy.  

Fig.~\ref{fig:whitematter}(c) compares changes in ensemble spectral energy across the five defined bands during an fMRI task. SE in band 1 decreases over time, while bands 2-4 show an overall increase, and band 5 remains largely unchanged. Notably, spectral content shifts significantly 10 seconds post-stimulus, potentially reflecting delayed hemodynamic response peaks in white matter, observed in multiple WM fiber bundles at a similar delay~\citep{Li2019}.

\section{Conclusion \& Outlook}
We presented a scheme for designing signal-adapted frames on graphs, leveraging the ensemble energy spectral density of a given signal set. The design, based solely on stationary signal information, allows flexibility in capturing non-stationary features through adjustable bandpass characteristics. While formulated on the graph Laplacian spectrum, the design can be extended to the spectrum of other graph shift operators, e.g. the graph adjacency spectrum to enable signal-adapted decomposition on directed graphs~\citep{Sevi2023, Chan2024, Stankovic2025}, or even to that of high-order networks~\citep{Schaub2021, Santoro2023, DalCol2024}. 
The proposed methods may also be used to design efficient graph filters for use in graph convolutional layers~\citep{Defferrard2016, Dong2020, Liu2024}. 

\section*{Acknowledgements}
This work draws in part on material previously published in \citep{Abramian2021, Behjat2021isbi, Behjat2016, Milloz2024biorxiv}.

\section*{Appendix 1 - Proof of Proposition~\ref{propo:BsplineFrame}}

The sum of squared magnitudes of B-spline based spectral kernels $\{B_{j}(\lambda)\}_{j=1}^{J}$ forms a partition of unity since
\begin{align}
\sum_{j=1}^{J} |B_{j}(\lambda)|^{2} 
& \:\: \stackrel{(\ref{eq:B})}{=} \sum_{i=\Delta}^{J+\Delta+1} |\widetilde{B}_{i}(\lambda)|^{2} 
\nonumber
\\
& \:\: \stackrel{(\ref{eq:B_l})}{=} \sum_{i=\Delta}^{J+\Delta+1} \beta^{(n)}\left (\frac{\lambda_{\text{max}}}{J-1} (\lambda - i+1) \right )
\nonumber
\\
& \stackrel{i-1 \rightarrow k}{=} \sum_{k=\Delta-1}^{J+\Delta} \beta^{(n)}\left (\frac{\lambda_{\text{max}}}{J-1} (\lambda - k) \right )
\nonumber
\\
& \quad = 1
\nonumber.
\end{align}
where in the last equality we use the property that integer shifted splines form a partition of unity. 

\section*{Appendix 2 - Proof of Proposition~\ref{propo:umt}}

In order to ensure that the spectral kernels cover the full spectrum, $a$ must be chosen such that 
\begin{equation*}
\lambda_{\text{max}}   \stackrel{(\ref{eq:kernelScaleJ})}{=} \lambda_{\text{II}}+a 
\stackrel{(j=J)}{=} 
\gamma a+(J-2)\Delta + a, 
\end{equation*}
which using (\ref{eq:Delta}) leads to $ a = \frac{\lambda_{\text{max}}}{J\gamma -J -\gamma +3}.$

To prove that the UMT system of spectral kernels form a tight frame, (\ref{eq:tightFrame}) needs to be fulfilled. Since, for all j, the supports of $U_{j-1}(\lambda)$ and $U_{j+1}(\lambda)$ are disjoint, $G(\lambda)$ can be determined as 
\begin{align}
\label{eq:pouUMT}
G(\lambda) & = \sum_{j=1}^{J} |U_{j}(\lambda)|^{2}  \nonumber  \\
\nonumber 
& \: \: \stackrel{(\ref{eq:uniformFrame})}{=} 
\begin{cases}
|U_{1}(\lambda)|^{2} \: \stackrel{(\ref{eq:kernelScale1})}{=} 1 & \forall \lambda \in [0,a]
\\  |U_{1}(\lambda)|^{2} + |U_{2}(\lambda)|^{2} & \forall \lambda \in ]a,\gamma a]
\\  |U_{2}(\lambda)|^{2} + |U_{3}(\lambda)|^{2} & \forall \lambda \in ]\gamma a,\gamma a + \Delta]
\\ \vdots & \vdots
\\ |U_{J}(\lambda)|^{2} \: \stackrel{(\ref{eq:kernelScaleJ})}{=} 1 & \forall \lambda \in ]\lambda_{\text{max}}-a,\lambda_{\text{max}}]
\end{cases} \\
& \stackrel{(\ref{eq:kernelScalej})}{=} 
\begin{cases}
1 & \forall \lambda \in [0,a]
\\ \cos^{2}(x_{\text{I}}) + \sin^{2}(x_{\text{I}}) & \forall \lambda \in ]a,\gamma a]
\\ \cos^{2}(x_{\text{II}}) + \sin^{2}(x_{\text{II}}) & \forall \lambda \in ]\gamma a,\gamma a + \Delta]
\\ \vdots & \vdots
\\ 1 & \forall \lambda \in ]\lambda_{\text{max}}-a,\lambda_{\text{max}}]
\end{cases} \nonumber \\
& \: \: = 1  \quad \forall \lambda \in [0,\lambda_{\text{max}}]
\end{align} 
where $x_{\text{I}} = \frac{\pi}{2} \nu (\frac{1}{\gamma-1}(\frac{\lambda}{a}-1)) $ and $x_{\text{II}} = \frac{\pi}{2} \nu (\frac{1}{\gamma-1}(\frac{\lambda - \Delta}{a}-1))$. 

For any given $\gamma$, the constructed set of spectral kernels form a tight frame. However, in order for the frame to satisfy the uniformity constraint given in (\ref{eq:uniformityConstraintContinious}), the appropriate $\gamma$ needs to be determined. From (\ref{eq:kernelScalej}), we have $\forall j \in  \{2, \ldots ,J-2\}$   
\begin{align}
\label{eq:A1}
 U_{j}(\lambda) & = U_{j+1}(\lambda + \Delta) \quad   \:  \forall \lambda \in ]\lambda_{\text{I}},\lambda_{\text{II}} + \Delta].
\end{align}
By considering an inverse linear mapping of the spectral support where $U_{1}(\lambda) \ne 0$, i.e. $[0, \gamma a]$, to the spectral support where $U_{J}(\lambda) \ne 0$, i.e. $[\lambda_{\text{max}}-\gamma a,\lambda_{\text{max}}]$, we have
\begin{equation}
\label{eq:A2}
U_{1}(\lambda)  = U_{J}(- \lambda + 2a + J \Delta)  \quad   \:  \forall \lambda \in [0,\gamma a].
\end{equation}
Thus, from (\ref{eq:A1}) and (\ref{eq:A2}) we have     
\begin{subequations}
\begin{align}
\int_{0}^{\lambda_{\text{max}}} {U_{j}}(\lambda) d\lambda  & =  C_2,  \quad j=2,\ldots,J-1  \\
\int_{0}^{\lambda_{\text{max}}} U_{1}(\lambda) d\lambda  & = \int_{0}^{\lambda_{\text{max}}} {U_{J}}(\lambda) d\lambda  = C_1,  
\end{align}
\end{subequations}
respectively, where $C_1,C_2 \in \mathbb{R}^{+}$. Thus, in order to satisfy (\ref{eq:uniformityConstraintContinious}), $\gamma$ should be chosen such that 
\begin{align}
C_1 & = C_2 \nonumber
\\\int_{0}^{\lambda_{\text{max}}} U_{1}(\lambda) d\lambda  & = \int_{0}^{\lambda_{\text{max}}} {U_{2}}(\lambda) d\lambda  \nonumber
\\
a + \int_{a}^{\gamma a} U_{1}(\lambda) d\lambda  & =  \int_{a}^{\gamma a} \sin(\frac{\pi}{2} \nu(\frac{1}{\gamma -1} (\frac{\lambda}{a} -1))) d\lambda  \nonumber
\\ & + \int_{\gamma a}^{\gamma a+\Delta} {U_{2}}(\lambda) d\lambda \nonumber
\\ \label{eq:lOfGamma} a & \stackrel{(\ref{eq:A1})}{=}  \int_{a}^{\gamma a} \sin(\frac{\pi}{2} \nu(\frac{1}{\gamma -1} (\frac{\lambda}{a} -1))) d\lambda.
\end{align} 
The optimal $\gamma$ that satisfies (\ref{eq:lOfGamma}) was obtained numerically by defining 
\begin{equation}
\label{eq:closedForm}
Q(\gamma ) = \int_{a}^{\gamma a} \sin(\frac{\pi}{2} \nu(\frac{1}{\gamma -1} (\frac{\lambda}{a} -1))) d\lambda - a, 
\end{equation} 
and discretizing $Q(\gamma )$ within the range $(a,\gamma a]$, with a sampling factor of $1 \times 10^{-4}$. Testing for $\gamma \ge 1$, with a step size of $1 \times 10^{-2}$, the optimal value, which is independent of $\lambda_{\text{max}}$ and $J$, was found to be $\gamma  = 2.73$.

\section*{References}
\bibliographystyle{bxv_abbrvnat}
\bibliography{main}
\end{document}